\documentclass[conference]{IEEEtran}
\usepackage{algorithm}
\usepackage{algorithmic}
\usepackage{amsmath,amssymb,amsfonts,amsthm}
\usepackage{ascmac}
\usepackage{bm}
\usepackage{cite}
\usepackage{footnote}
\usepackage{graphicx}
\usepackage{mathrsfs}
\usepackage{multirow}
\usepackage[subrefformat=parens]{subcaption}
\usepackage{textcomp}
\usepackage{threeparttable}
\usepackage[normalem]{ulem}
\usepackage{url}
\usepackage{wasysym}
\usepackage{xcolor}
\usepackage{xspace}

\usepackage{algorithm}
\usepackage{algorithmic}

\newcommand{\Gt}{G_\mathsf{t}}
\newcommand{\Vt}{V_\mathsf{t}}
\newcommand{\vt}{v_\mathsf{t}}
\newcommand{\Et}{E_\mathsf{t}}
\newcommand{\et}{e_\mathsf{t}}
\newcommand{\advG}{G^\mathsf{adv}}
\newcommand{\xe}{\mathbf{x}_e}
\newcommand{\xet}{\mathbf{x}_{\et}}
\newcommand{\TCD}{\mathsf{TCD}}
\newcommand{\TCDalpha}{\mathsf{TCD}_\alpha}
\newcommand{\TCDone}{\mathsf{TCD}_1}
\newcommand{\TCDtwo}{\mathsf{TCD}_2}
\newcommand{\TCDmax}{\mathsf{TCD}_\infty}
\newcommand{\TTCD}{\mathsf{TTCD}}
\newcommand{\Dtrain}{D_\mathsf{train}}
\newcommand{\fr}{f_\mathsf{r}}

\newcommand{\bbR}{\mathbb{R}}
\newtheorem{definition}{Definition}
\newtheorem{theorem}{Theorem}
\newtheorem{lemma}{Lemma}

\newif\ifconferenceon\conferenceontrue
\ifconferenceon
\newcommand{\arxiv}[1]{}
\newcommand{\conference}[1]{#1}
\else
\newcommand{\arxiv}[1]{#1}
\newcommand{\conference}[1]{}
\fi

\allowdisplaybreaks[1]

\begin{document}

\title{R-HTDetector: Robust Hardware-Trojan Detection Based on Adversarial Training}

\author{
    \IEEEauthorblockN{
        Kento Hasegawa\IEEEauthorrefmark{1}\textsuperscript{\textsection},
        Seira Hidano\IEEEauthorrefmark{1}\textsuperscript{\textsection},
        Kohei Nozawa\IEEEauthorrefmark{2},
        Shinsaku Kiyomoto\IEEEauthorrefmark{1}, and\\
        Nozomu Togawa\IEEEauthorrefmark{2}
    }
    \IEEEauthorblockA{
        \IEEEauthorrefmark{1}KDDI Research, Inc.,
        \IEEEauthorrefmark{2}Waseda University
        % Email: \{kt-hasegawa, se-hidano\}@kddi-research.jp
    }
}
% \author{
%     ~\\
%     ~
% }

\maketitle

% Show equal contribution remark at the footnote
\begingroup\renewcommand\thefootnote{\textsection}
\footnotetext{Equal contribution}
\endgroup

\begin{abstract}
    Hardware Trojans~(HTs) have become a serious problem, and extermination of them is strongly required for enhancing the security and safety of integrated circuits. An effective solution is to identify HTs at the gate level via machine learning techniques. However, machine learning has specific vulnerabilities, such as \emph{adversarial examples}. In reality, it has been reported that adversarial modified HTs greatly degrade the performance of a machine learning-based HT detection method. Therefore, we propose a robust HT detection method using adversarial training (\emph{R-HTDetector}). We formally describe the robustness of R-HTDetector in modifying HTs. 
    Our work gives the world-first adversarial training for HT detection with theoretical backgrounds. We show through experiments with Trust-HUB benchmarks that R-HTDetector overcomes adversarial examples while maintaining its original accuracy.
\end{abstract}


\section{Introduction}\label{sec:introduction}

% \section{Introduction}
% \label{sec:introduction}
The increase in hardware devices expands the demand for integrated circuits (ICs).
The outsourcing of IC design and manufacturing to third parties is expected to realize the efficient development of IC products.
When third parties are involved in a complicated hardware supply chain, there is apprehension that malicious circuits, which are referred to as \emph{hardware Trojans}~(HTs), are inserted into products~\cite{Francq2015,Xiao2016}.
HTs cause serious security and safety problems, such as information leakage to the outside, performance degradation, and suspension of operations.
In particular, numerous third parties' modules and intellectual properties (IPs) are incorporated in the design process. 
Thus, developing detection techniques that can be applied to IC design information is strongly needed.
Fig.~\ref{fig:detection_scenario} shows a scenario of HT insertion and detection in the design process.

Several HT detection methods for IC design processes focus on gate-level circuits~\cite{Oya2015, Hasegawa2017a,Wang2019b,9338932,7577739,8702462}.
It was emphasized that the nets that compose a hardware Trojan (Trojan nets) have specific features~\cite{Oya2015}.
% Since a large circuit contains more than millions of nets, it is desired to find Trojan nets with those features efficiently.
Since the size of a large circuit exceeds millions of nets, it is necessary to efficiently identify Trojan nets with these features.
Machine learning techniques have been introduced to HT detection based on structural features at the gate level~\cite{Hasegawa2017,Wang2019b,9338932}.
It has been experimentally demonstrated that machine learning-based HT detection methods achieve high detection accuracy by learning sufficient quality training data.
However, the vulnerability of HT detection models has not been sufficiently discussed.
Specifically, there is growing concern about the security of machine learning.
\emph{Adversarial examples} are one of the most serious threats to machine learning~\cite{Goodfellow2015}.
These examples are generated by adding small perturbations to an input example, resulting in the misclassification of a target classifier.
Various types of adversarial examples have been developed, and defending machine learning-based systems from adversarial examples becomes a significant challenge to be solved~\cite{Akhtar2018}.

\begin{figure}[t]
    \centering
    \includegraphics[width=1.0\columnwidth]{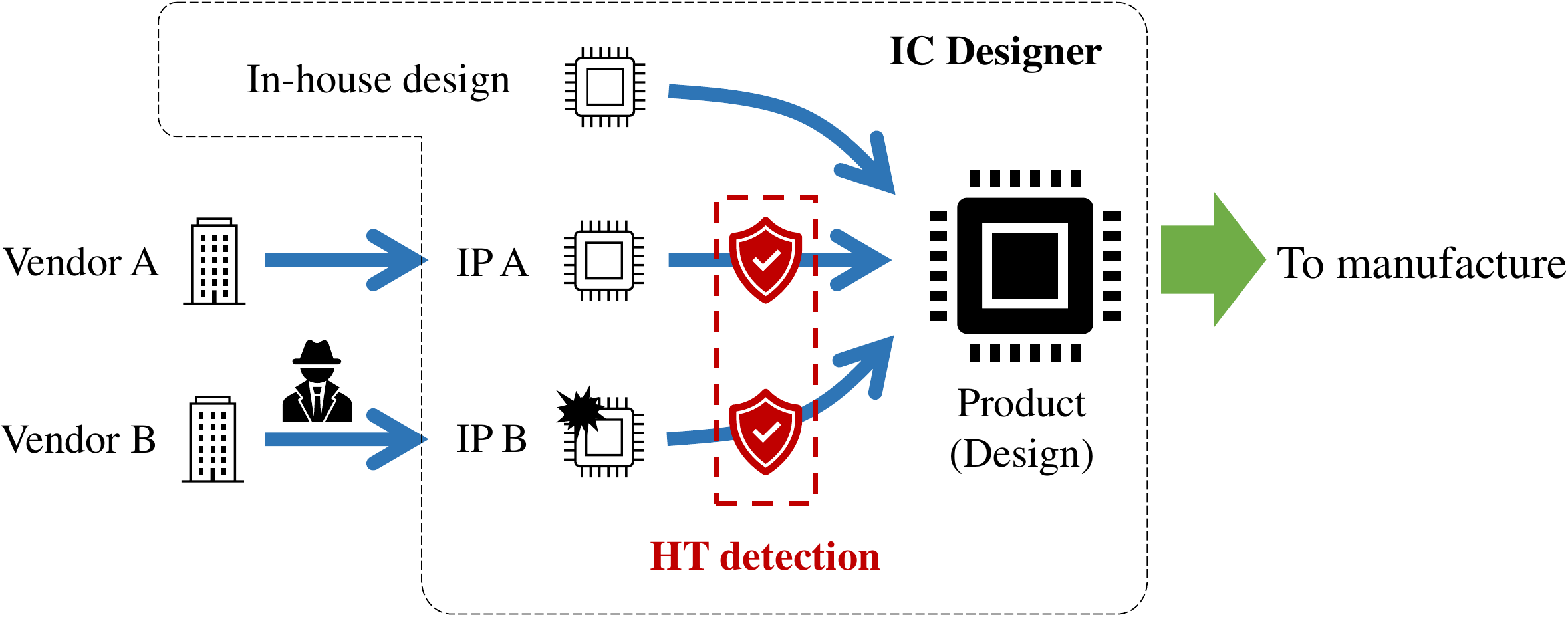}
    \caption{Scenario of HT detection in the IC design process. A designer incorporates IPs A and B provided by vendors A and B, respectively, into her IC product. HT detection is performed at the point in a red circle. Even though an HT is inserted, such as in IP B, the designer can effectively identify the HT.}
    \label{fig:detection_scenario}
\end{figure}
   
Existing HT detection methods have been designed without any consideration of adversarial examples~\cite{Liakos2020}.
These methods can be affected by adversarial examples.
A recent study~\cite{Nozawa2019a} proposed a framework for generating adversarial examples for HT detection with machine learning.
Adversarial examples for HT detection are generated by modifying a small number of gates that compose the HT.
This type of attack is referred to as a \emph{gate modification attack}.
It was also demonstrated that this attack degrades the performance of detection methods with a neural network.
As shown in Fig.~\ref{fig:detection_scenario}, malicious vendors are located on the outside of the organization of IC designers.
In this case, malicious vendors who know the detection method can modify their circuits and provide the modified circuits to IC designers.  
Therefore, gate modification attacks are realistic, and strong HT detection is required to overcome them.

\smallskip
\noindent\textbf{Contributions.}
In this paper, we propose a robust HT detection method (\emph{R-HTDetector}).
R-HTDetector is based on \emph{adversarial training}, which is a learning method that makes a classifier robust to adversarial examples.
Our main contributions are summarized as follows:
\begin{itemize}
%  \item We generalize gate modification attacks proposed in \cite{Nozawa2019a}.
%  R-HTDetector aims for robustness to gate modification attacks with various purposes.
%  However, the attack with TCD targets only a certain purpose of the adversary. 
%  We thus introduce a new metric called \emph{$\alpha$-TCD} for realizing attacks with different purposes.
 \item We generalize gate modification attacks and propose a metric named \emph{$\alpha$-TCD}, which effectively attacks machine learning-based HT detection.
 \item We design adversarial training for HT detection and propose a robust HT detection method named R-HTDetector.
 Our proposed method is based on adversarial examples generated with a metric referred to as \emph{targeted TCD} (\emph{TTCD}).
 \item We theoretically establish that adversarial training with TTCD overcomes any gate modification attack with $\alpha$-TCD.
%  As far as we know, there exists no work to formally describe the robustness of adversarial training to HTs.
 There exists no work to formally describe the robustness of adversarial training to HTs.
 Our work gives the world-first adversarial training on HT detection with theoretical backgrounds.
 \item We experimentally demonstrate with HT benchmarks that attacks with $\alpha$-TCD substantially degrade the performance of the neural network-based detection method~\cite{Hasegawa2017}.
 The attacks achieve a maximum decrease in \emph{true positive rates} (\emph{TPRs}) of 61.8\% using the Trust-HUB benchmark~\cite{Salmani2013,shakya2017benchmarking}.
 \item We also show that R-HTDetector can identify Trojan nets with high accuracy even though the HTs are adversarially distorted.
 R-HTDetector achieves average TPRs of grater than 78\% over the Trust-HUB benchmark for all attacks.
 Since the average TPR for the original Trojan nets is 75.2\%, this result indicates that R-HTDetector is robust to both original Trojans and adversarially altered Trojans.
 We further demonstrate that R-HTDetector can scale to another circuit design using the TRIT-TC benchmark.
\end{itemize}

\section{Backgrounds}
\label{sec:preliminaries}

This section presents the background and notations of the proposed method.

\subsection{HT Detection and Machine Learning}
\label{sub:htdetection}

An HT is a malicious circuit that is unintentionally inserted by an adversary.
In general, an HT is composed of a trigger circuit and payload circuit.
The trigger circuit monitors the internal signals to determine whether the trigger condition is satisfied.
The payload circuit is the malicious function that an adversary wishes to perform on the IC product.
The adversary carefully sets the trigger condition such that the HT is difficult to detect during testing or ordinal usage.

As shown in Fig.~\ref{fig:detection_scenario}, the scenario addressed in this paper is that an IC product is composed of several third-party IPs and that an untrusted third-party vendor provides an HT-infested IP to IC designers.
To ensure that HTs are not inserted into the IC design, IC designers need to examine the design.
An effective way to remove the threat of HTs from IC design is to detect them at early stages of development.
From the perspective of earlier-stage detection, gate-level design is an attractive target.
Gate-level design information involves the logic elements employed in a circuit and represents the structure of the circuit.
% Detecting HTs at early steps of the development is an effective way to remove the threats of them.
An early study~\cite{Oya2015} assigned a score that represents the likelihood of an HT to each net that composes a circuit. 
In particular, the trigger circuit that identifies complex trigger conditions is specific to HTs.
The structural features of trigger circuits and the combination of trigger circuits and typical payload circuits are useful for HT detection.
However, the structural features of HTs are manually extracted from benchmark netlists, and the threshold scores are carefully designed in \cite{Oya2015}.
It is difficult to quickly update the characteristic structures and scores for a novel HT.

Machine learning techniques effectively learn the structural features of an HT at the gate level~\cite{Hasegawa2017}.
Specific features are designed to effectively identify HTs, and a machine learning-based HT detection method with these features could detect Trojan nets with an 84.8\% detection rate~\cite{Hasegawa2017}.
% Since this method is a basic approach, we develop our proposed method based on the method in \cite{Hasegawa2017}.
Table~\ref{tb:11_features} lists the 51 features that represent each net $e$ in \cite{Hasegawa2017}.
The features represent the number of specific circuit elements or the minimum distance to specific circuit elements from the target net.
% 51 features are introduced in \cite{Hasegawa2017} for effectively identifying Trojan nets.
% Table \ref{tb:11_features} shows representative 11 features among all the features.
In \cite{Hasegawa2017}, 11 features are selected as a set of the most effective features for HT detection using a random forest classifier.
Although the effective features may depend on the training datasets, the structural features listed in Table~\ref{tb:11_features} help detect HTs from a gate-level netlist using machine learning.

% \subsubsection{Notations}
\smallskip
\noindent
\textbf{Notations.} A gate-level netlist can be modeled as graph $G$, as discussed in \cite{8913629}, in which a node represents a gate and an edge represents a net.
$G$ consists of multiple gates $v \in V$ and nets $e \in E$.
If Trojan circuit $\Gt = (\Vt, \Et)$ is embedded in $G = (V, E)$, $V$ and $E$ include Trojan gates $\vt \in \Vt$ and nets $\et \in \Et$, respectively.
% In this paper, we focus on a machine-learning-based HT detection method~\cite{Hasegawa2017}.
The method in \cite{Hasegawa2017} predicts whether a given net $e \in E$ is a Trojan net $\et \in \Et$ with a neural network.
Let $\xe \in \bbR^{n}$ denote an $n$-dimensional feature vector that represents net $e \in E$.
The prediction is performed with a neural network-based detection model $f:\bbR^{n} \rightarrow [0, 1]$ that maps feature vector $\xe$ to the probability that the corresponding net $e$ is a Trojan net.
The true positive rate~(TPR) and true negative rate~(TNR) are employed to evaluate the HT detection performance of the target model $f$.
Let $Y_\mathsf{t}$ (resp. $Y_\mathsf{n}$) be the set of nets predicted to be Trojan nets (resp. normal nets).
TPR and TNR are expressed as follows:
$\mathrm{TPR} = {|\Et \cap Y_\mathsf{t}|} / {|\Et|}$ and $\mathrm{TNR} = {|(E \setminus \Et) \cap Y_\mathsf{n}|} / {|E \setminus \Et|}$.
% \begin{align}
%   \mathrm{TPR} & = {|\Et \cap Y_\mathsf{t}|} / {|\Et|} \label{eq:tpr} \\
%   \mathrm{TNR} & = {|(E \setminus \Et) \cap Y_\mathsf{n}|} / {|E \setminus \Et|} \label{eq:tnr}
% \end{align}
In particular, TPR (also known as recall) is a significant metric because we aim to catch as many Trojan nets as possible.

\begin{table}[t]
\caption{51 features that represent net $e$ presented in \cite{Hasegawa2017}.}
\label{tb:11_features}
\centering
\scalebox{1.00}{
\begin{tabular}{c|p{70mm}}
\hline
\# & Description of Feature \\
\hline
1--5 & Number of logic-gate fanins $n$-level away from $e$ ($1 \le n \le 5$).\\
6--10 & Number of flip-flops up to $n$-level away from the input side of $e$ ($1 \le n \le 5$).\\
11--15 & Number of flip-flops up to $n$-level away from the output side of $e$ ($1 \le n \le 5$).\\
16--20 & Number of multiplexers up to $n$-level away from the input side of $e$ ($1 \le n \le 5$).\\
21--25 & Number of multiplexers up to $n$-level away from the output side of $e$ ($1 \le n \le 5$).\\
26--30 & Number of up to $n$-level loops from the input side of $e$ ($1 \le n \le 5$).\\
31--35 & Number of up to $n$-level loops from the output side of $e$ ($1 \le n \le 5$).\\
36--40 & Number of constants (fixed at the high or low level) up to $n$-level away from the input side of $e$ ($1 \le n \le 5$).\\
41--45 & Number of constants (fixed at the high or low level) up to $n$-level away from the output side of $e$ ($1 \le n \le 5$).\\
46 & Minimum level to the primary input from $e$.\\
47 & Minimum level to the primary output from $e$.\\
48 & Minimum level to any flip-flop from the input side of $e$.\\
49 & Minimum level to any flip-flop from the output side of $e$.\\
50 & Minimum level to any multiplexer from the input side of $e$.\\
51 & Minimum level to any multiplexer from the output side of $e$.\\
\hline
\end{tabular}
}
\end{table}

\subsection{Gate Modification Attacks}
\label{sub:gm_attacks}
A recent study~\cite{Nozawa2019a} on HT detection proposed a framework for generating \emph{adversarial examples} against machine learning-based detection methods.

Adversarial examples and attacks using them are actively investigated as an emerging theme in AI security~\cite{10.1145/3374217}.
In image processing, such an attack is launched by adding imperceptible noise (also known as perturbation) to original images.
Carefully crafted perturbation fools a target model, resulting in misclassification.
To examine the robustness of machine learning models, the research of adversarial examples on specific applications has been expanded.

In \cite{Nozawa2019a}, the adversary adversarially modifies gates such that target detection model $f$ misclassifies Trojan nets that composes her Trojan circuit.
The attack is referred to as a \emph{gate modification attacks}.
Unlike images, circuits have specific constraints.
For instance, the modified circuit has to operate correctly.
Adding small noise (or applying small change) to the circuit design may destroy the original functionality of the product and/or HT, and consequently, testers or users notice the presence of unintended modification of the circuit.
In this sense, general attack methods that use adversarial examples cannot be applied to gate-level netlists.
The work in \cite{Nozawa2019a} thus introduced logically equivalent modification patterns.
Logically equivalent modification does not break the original functionality of the target circuit.
Therefore, once adversaries design HTs, they can easily generate variants.
% Specifically, six modification patterns were provided (\textit{m1}--\textit{m5} and \textit{m14} in Fig.~\ref{fig:modification_patterns} at Section~\ref{sub:setup} are proposed in \cite{Nozawa2019a}).
As listed in Table~\ref{tb:11_features}, the structural feature-based HT detection observes the number of neighbor circuit elements and the minimum distance to specific circuit elements.
If the structure of the circuit is changed such that the altered circuit is logically equivalent, the feature values provided to a machine learning-based HT detection model are changed, and the small change may fool the model.
Logically equivalent modification is one of the most promising ways to automatically generate variants for the purpose of adversarial example attacks.

% A large number of changes leads to easy detection of modified circuits.
A large number of changes facilitates the detection of modified circuits.
To address this problem, \cite{Nozawa2019a} provided a metric for maximizing the number of misclassified Trojan nets with a small number of modifications.
The metric is referred to as the \emph{Trojan-net concealment degree} (\emph{TCD}).
The TCD with respect to Trojan nets $\Et$ and to detection model $f$ is represented by
\begin{align}
\TCD(\Et, f) = \frac{1}{|\Et|} \sum_{\et \in \Et} \log f(\xet).
\label{eq:tcd}
\end{align}
As described in Section~\ref{sec:preliminaries}, the detection performance of a model $f$ is evaluated as $\mathrm{TPR} = {|\Et \cap Y_\mathsf{t}|} / {|\Et|}$.
Minimizing $\sum_{\et \in \Et} \log f(\xet)$ decreases $|\Et \cap Y_\mathsf{t}|$, and thus, the detection model deteriorates.
The adversary can conceal many Trojan nets by selecting gates such that the TCD for the modified Trojan circuit is minimized.
However, since this optimization problem is difficult to solve, an alternative solution is developed in \cite{Nozawa2019a}.
In \cite{Nozawa2019a}, a gate that minimizes the metric is selected at a certain time, and this process is repeated at most $K$ times, i.e., $K$ gates are modified in total.

\subsection{Adversarial Training}
\label{sub:adversarial_training}
Adversarial training aims to make a classifier robust to adversarial examples~\cite{Goodfellow2015,Goodfellow2017,DBLP:conf/iclr/KurakinGB17}.
In adversarial training, adversarial examples are newly generated at each iteration, and added to the training dataset.
The training iteration is repeatedly performed a certain number of times.
% Adversarial training is well-studied in the field of image classification and also has recently been applied to various fields~\cite{Goodfellow2017,Dai2018}.
In the adversarial training of image classification, an adversarial image is generated by manipulating an original image on a pixel-by-pixel basis.

When we consider adding the perturbation to the HTs, this process corresponds to manipulating a netlist on a net-by-net basis.
However, adversarial examples for HT detection are generated on a circuit-by-circuit basis because modifying a net affects the whole circuit~\cite{Nozawa2019a}.
Thus, it is not possible to simply apply typical adversarial training to HT detection.

\section{Threat Model}
\label{sec:threat_model}
This section presents the threat model for HT detection with machine learning.

As illustrated in Fig.~\ref{fig:detection_scenario}, HTs at the gate level are incorporated in an IC product during the design process.
% A circuit consisting of multiple modules and IPs is sufficiently large, and it becomes difficult to detect the HTs.
The ideal way to exterminate all the HTs is to detect them and remove the compromised modules or IPs before they are integrated into the original design.
Thus, HT detection must be performed on each module or IP during the design process (e.g., a commercial service is available in \cite{htfinder}).
However, the adversary may be able to use the HT detection system as well.
In this paper, it is assumed that the adversary modifies her own modules and IPs to avoid detection using the detection system.
Specifically, we assume gray-box access to detection model $f$ via the system as follows:
\begin{itemize}
\item The adversary can input any Trojan net $e_\mathsf{t}$ to detection model $f$ multiple times and obtain the output of $f$ for the given net $e_\mathsf{t}$ (i.e., probability that $e_\mathsf{t}$ is a Trojan net).
\item The adversary cannot directly obtain the structure and parameters of detection model $f$.
\end{itemize}

The detection system provides not only the conclusive predicted result (i.e., normal or Trojan) but also the output of $f$ (i.e., probability that a given net is a Trojan net).
The trustworthiness of machine learning is considered, and existing policies~\cite{EC2019,OECD2019} have regularized the transparency of decision-making by machine learning.
To accomplish this objective, machine learning-based systems are required to present the reason for the model's decision, and the output of $f$ can be employed for such a purpose.
The gray-box access to $f$ must be reasonable enough.

\section{Proposed Method}
\label{sec:method}
In this section, we propose an HT detection method named \emph{R-HTDetector} that is robust in responding to gate modification attacks.
First, an overview of the proposed method is presented.
The proposed method consists of two parts: the gate modification attack and adversarial training.
% In the gate modification attack, a method of generating adversarial examples of gate-level netlist is proposed.
% To effectively generate adversarial examples, the metrics evaluating the efficiency of the example is generalized as \emph{$\alpha$-TCD}.
% Using the framework of the gate modification attack, adversarial training method is proposed to mitigate the attack.
% In adversarial training, the small number of training samples are replace with their adversarial examples.
% By updating the weight of the target model using the replaced training dataset, a robust model is constructed.
% At the end of this section, the effectiveness of the proposed method is proven via formal analysis.
At the end of this section, the effectiveness of the proposed method is shown via formal analysis.
% We first generalize gate modification attacks by introducing a new metric called \emph{$\alpha$-TCD}.
% We then develop an adversarial training method that mitigates the attacks using $\alpha$-TCD comprehensively.
% We finally prove the effectiveness of our proposed method via formal analysis.

\begin{figure*}[t]
  \centering
  \includegraphics[width=0.85\linewidth]{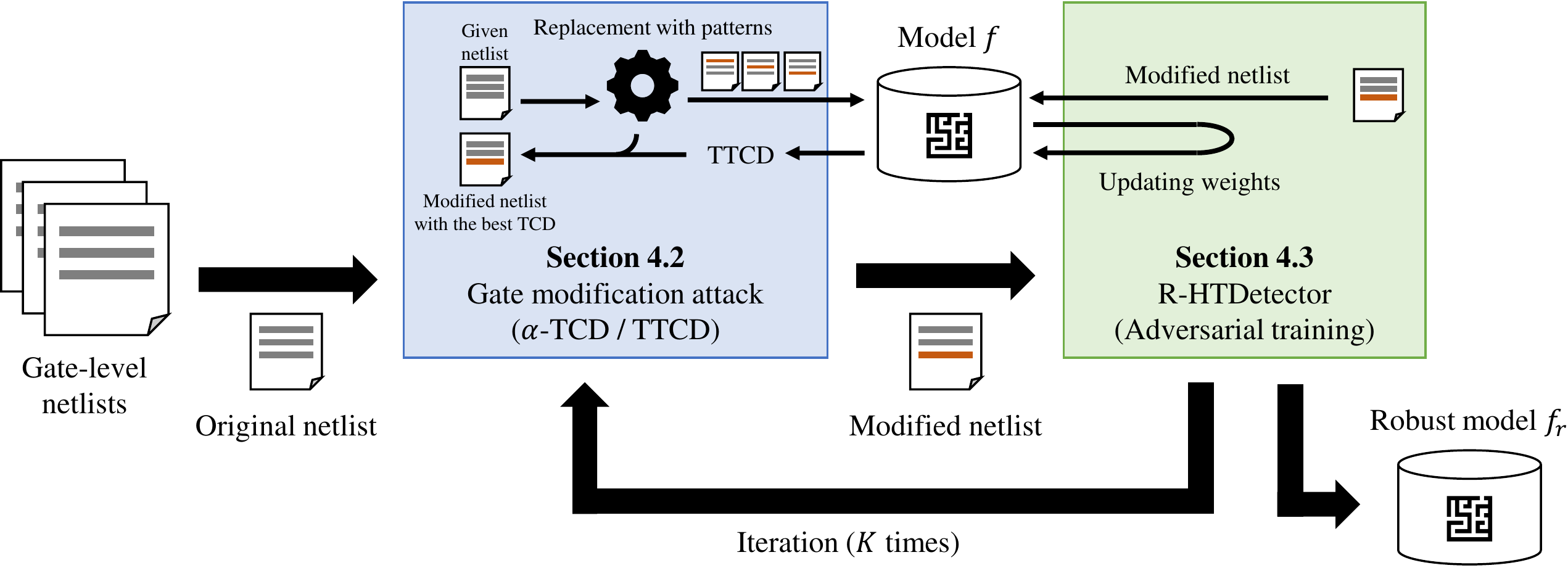}
  \caption{Overview of the proposed method. The proposed method consists of two parts: gate modification attack and adversarial training. In the gate modification attack, adversarial examples on gate-level netlists are generated. In adversarial training, the target model learns the generated adversarial examples, and a robust model is constructed.}
  \label{fig:rhtdetector-overview}
\end{figure*}

\subsection{Overview of the Proposed Method}

Fig.~\ref{fig:rhtdetector-overview} shows an overview of the proposed method.

The gate modification attack method is performed to generate adversarial examples of a given gate-level netlist.
In the gate modification attack method, several gates in the given gate-level netlist are replaced with another set of logically equivalent gates, and adversarial examples are generated.
To effectively generate adversarial examples, the metrics that evaluate the efficiency of the example are generalized as \emph{$\alpha$-TCD}.
The generated examples are provided to target model $f$, and $\alpha$-TCD values are obtained.
The example with the smallest $\alpha$-TCD is selected as the most effectively modified sample.
For adversarial training, another metric is introduced instead of $\alpha$-TCD.

Adversarial training in the proposed method is the framework in which the target model is enhanced by learning the samples generated by the gate modification attack.
% Next, the adversarial training method is established using the framework of the gate modification attack.
The goal of the method is to construct a robust model against adversarial examples.
The new metric named \emph{target TCD~(TTCD)} is adopted to generate adversarial examples.
By replacing the small number of training datasets with adversarial examples, a new training dataset is generated.
The target model is trained using the generated training dataset, and a robust model is constructed.

\subsection{Generalized Gate Modification Attacks}
\label{sub:alpha_gm_attacks}
Gate modification attacks are aimed at concealing as many Trojan nets as possible.
The detection model $f$ classifies net $e \in E$ as a normal net when the probability $f(\xe)$ is less than 0.5.
% However, the attack with TCD gives priority to reduce the probabilities of few nets as much as possible even though those values fall below 0.5.
% We thus introduce a generalized version of TCD for realizing the original purpose.
% We call this metric \emph{$\alpha$-TCD}, and define $\alpha$-TCD with respect to Trojan nets $\Et$ and a detection model $f$ as follows:
Gate modification attacks try to reduce the probability $f(\xe)$ of Trojan nets and to conceal them from the detection method.
To realize this purpose, we define $\alpha$-TCD with respect to Trojan nets $\Et$ and detection model $f$ as follows:
\begin{align}
% \TCDalpha(\Et, f) = \frac{1}{|\Et|} \sum_{\et \in \Et} [\log f(\xet)]^\alpha,
% \TCDalpha(\Et, f) = - \frac{1}{|\Et|} \sum_{\et \in \Et} \| \log f(\xet) \|_\alpha,
\TCDalpha(\Et, f) = - \frac{1}{|\Et|} \sum_{\et \in \Et} | \log f(\xet) |^\alpha,
% \TCDalpha(\Et, f) = - \frac{1}{|\Et|} \| \log f(\xet) \|_\alpha,
% \TCDalpha(\Et, f) = - \frac{1}{|\Et|} \| [ \log f(\xet) : \et \in \Et ] \|_\alpha,
\label{eq:alpha_tcd}
\end{align}
where $\alpha \geq 0$ is a positive parameter.
Equation~(\ref{eq:alpha_tcd}) can be derived from the cross-entropy of the model output $f(\xet)$ with respect to Trojan nets $\Et$.
% Equation~(\ref{eq:alpha_tcd}) shows the probability of being Trojan.
When $\alpha$-TCD is small, the genuine Trojan nets have a low probability of being Trojan for model $f$.
Therefore, if adversaries modify their HT so that $\alpha$-TCD is reduced, they can easily achieve their purpose to conceal Trojan nets.
The modification should be logically equivalent to maintaining the functionality of the original HT (regarding how to modify every gate; the details are presented in Section~\ref{sub:setup}).
Consequently, adversaries conceal the Trojan nets that are logically equivalent to the original Trojan nets, i.e., the gate modification attack is accomplished.
Note that since $\alpha$-TCD is estimated from only model outputs, our attacks can be applied to any type of detection model constructed with machine learning.

$\alpha$-TCD can also be interpreted as the $L^p$ norm.
Since $L^p$ norms for $p =$ 1, 2, and $\infty$ are well-known, we especially treat the attacks at $\alpha =$ 1, 2, and $\infty$ in this paper.
% \footnote{%
% When $\alpha = 1$, $\alpha$-TCD corresponds to the metric proposed in \cite{Nozawa2019a}.
% }
When $\alpha = 1$ or $2$, $\alpha$-TCD shows how likely the whole set of nets in $E_t$ is Trojan.
As expressed in Equation~(\ref{eq:alpha_tcd}), $\alpha$-TCD is derived based on the whole set of Trojan nets.
Specifically, $\alpha$-TCD at $\alpha = 1$ and $2$ are expressed as follows:
\begin{align}
  \TCDone(\Et, f) & = - \frac{1}{|\Et|} \sum_{\et \in \Et} | \log f(\xet) |, \label{eq:alpha_tcd_one} \\
  \TCDtwo(\Et, f) & = - \frac{1}{|\Et|} \sum_{\et \in \Et} | \log f(\xet) |^2. \label{eq:alpha_tcd_two}
\end{align}
If $\alpha$-TCD at $\alpha = 1$ or $2$ is decreased, the whole set of Trojan nets in $E_t$ is likely to be misclassified as normal.
When $\alpha = \infty$, $\alpha$-TCD shows the probability of the net with the highest probability of being Trojan in $E_t$.
% $\alpha$-TCD at $\alpha = \infty$ can be written as follows:
$\alpha$-TCD at $\alpha = \infty$ is expressed as follows:
% $\alpha$-TCD at $\alpha = 1, 2$, and $\infty$ can be written as follows:
%
% \begin{align}
% \TCDmax(\Et, f) = \max_{\et \in \Et} \log f(\xet).
% \label{eq:alpha_tcd_max}
% \end{align}
\begin{align}
  % \TCDone(\Et, f) & = - \frac{1}{|\Et|} \sum_{\et \in \Et} | \log f(\xet) |, \label{eq:alpha_tcd_one} \\
  % \TCDtwo(\Et, f) & = - \frac{1}{|\Et|} \sqrt{\sum_{\et \in \Et} | \log f(\xet) |^2}, \label{eq:alpha_tcd_two} \\
  \TCDmax(\Et, f) & = \max_{\et \in \Et} \log f(\xet) .  \label{eq:alpha_tcd_max}
\end{align}
The Trojan net with the largest cross-entropy $\log f(\xet)$ is recognized as the most likely Trojan net by a classifier.
If $\alpha$-TCD at $\alpha = \infty$ is decreased, the most likely Trojan net is misclassified as a normal net, and thus, the highest probability of being Trojan in $E_t$ is decreased.

% When the value of $\alpha$ is small, Trojan nets modified by the attack become hard to classify correctly, compared with large values of $\alpha$.
% In contrast, it is expected that many Trojan nets are misclassified by increasing the value of $\alpha$.

Algorithm~\ref{alg:gen_aes} summarizes gate modification attacks with $\alpha$-TCD.
This algorithm is a generalized version of that described in Section~\ref{sub:gm_attacks}.
The gates that compose the HT in a given circuit are replaced with another set of gates based on the modification patterns $P$ at lines 4 and 5.
The $\alpha$-TCD values are calculated at line 6, and the modified circuit with the smallest $\alpha$-TCD value is stored in $G^{i}$ at the $i$-th modification.
The modification process is repeated $K$ times, and the modified circuit is generated at line 14.

\begin{figure}[t]
  \begin{algorithm}[H]
  \caption{Gate modification attacks with $\alpha$-TCD.}
  \label{alg:gen_aes}
  \begin{algorithmic}[1]
  \REQUIRE Detection model $f$, circuit $G = (V, E)$ including HT $\Gt = (\Vt, \Et)$, modification patterns $P$, number of modified gates $K$
  \ENSURE Adversarial circuit $\advG$
  \STATE $i \leftarrow 0$, $G^{(i)} \leftarrow G$, $\mathsf{best\_TCD} \leftarrow 0$
  \WHILE{$i < K$}
    \FORALL{$\vt \in \Vt$ in $G^{(i)}$ that has not yet been modified}
    \FORALL{$p \in P$ that can be applied to $v_t$}
      \STATE $G' \leftarrow$ Apply $p$ to $\vt$ and generate the modified circuit.
      \IF{$\TCDalpha(G', f) < \mathsf{best\_TCD}$}
      \STATE $G^{(i)} \leftarrow G'$
      \STATE $\mathsf{best\_TCD} \leftarrow \TCDalpha(G', f)$
      \ENDIF
    \ENDFOR
    \ENDFOR
    \STATE $i \leftarrow i+1$, $G^{(i)} \leftarrow G^{(i-1)}$
  \ENDWHILE
  \RETURN $\advG = G^{(i)}$
  \end{algorithmic}
  \end{algorithm}
\end{figure}

\subsection{R-HTDetector: Adversarial Training for HT Detection}
\label{sub:rhtdector}
Adversarial training is employed to construct a detection model that is robust to gate modification attacks.
As described in Section~\ref{sub:adversarial_training}, adversarial training incorporates adversarial examples into a training dataset.
However, it is not practical to generate adversarial examples with $\alpha$-TCD for any value of $\alpha$.
To effectively apply adversarial training with a small number of adversarial examples to HT detection, \emph{weak adversarial examples} are considered in this paper.
Adversarial examples for HT detection are generated such that multiple Trojan nets are misclassified.
Each adversarial example is not optimal for a single Trojan net.
Non-optimal adversarial examples have a smaller perturbation in the feature space than the optimal example.
As a result, the attack performance for a target net is decreased.
We define such non-optimal adversarial examples as weak adversarial examples.
Intuitively, adversarial training based on adversarial examples with high attack performance renders a classifier robust not only to these examples but also to weak adversarial examples that are not involved in the training (formal discussion is presented in Section~\ref{sub:theoretical_analysis}).
Therefore, we propose an adversarial training method that is based on adversarial examples with a higher attack performance than any adversarial examples generated with $\alpha$-TCD.

We develop gate modification attacks that generate an optimal adversarial example for each Trojan net.
We introduce a new metric referred to as the \emph{targeted TCD} (\emph{TTCD}).
TTCD is defined for each Trojan net $\et \in \Et$ and is expressed as:
\begin{align}
\TTCD(\et, f) = \ \log f(\xet).
\label{eq:ttcd}
\end{align}
Equation~(\ref{eq:ttcd}) corresponds to the probability that the target net $\et$ is a Trojan net.

Adversarial training for HT detection generates adversarial examples via a gate modification attack with TTCD.
The algorithm of the attack with TTCD is almost the same as Algorithm~\ref{alg:gen_aes}.
However, each adversarial example is not generated every circuit but every Trojan net $\et \in \Et$.
Typical adversarial training generates adversarial examples for all labels, whereas the proposed method generates adversarial examples for the Trojan label.
% This is because the identification of Trojan nets is more important than that of normal nets in HT detection~\cite{Hasegawa2017,9338932}.
The identification of Trojan nets is more important than that of normal nets in HT detection~\cite{Hasegawa2017,9338932}.
Furthermore, the adversary can easily access the Trojan nets.
The adversary has a small motive to induce the misclassification of normal nets compared to Trojan nets.
Algorithm~\ref{alg:adversarial_training} summarizes the proposed adversarial training with TTCD, named \emph{R-HTDetector}, with respect to training dataset $\Dtrain$.
The training dataset is modified at the ratio of $l$.
The selected samples are modified by a gate modification attack with TTCD at line 6.
The mini-batch, including adversarial examples, is trained by model $f$ on line 9.
The update of the model weight is repeated $j$ times, and a robust model is constructed at line 12.

\begin{figure}[t]
    
  \begin{algorithm}[H]
    \caption{R-HTDetector: Adversarial training for HT detection.}
    \label{alg:adversarial_training}
    \begin{algorithmic}[1]
    \REQUIRE Training dataset $\Dtrain$, mini-batch size $m$, minimum Trojan samples in a mini-batch $m' (\leq m)$, ratio $l$ of Trojan samples to be modified, number of epochs $j$
    \ENSURE Robust detection model $\fr$
    \STATE $i \leftarrow 0$
    \STATE Initialize a detection model $f$.
    \WHILE{$i < j$}
      \STATE Sample a mini-batch $B$ of $m$ samples from $\Dtrain$.
    %   \STATE Make a mini-batch $B$ by randomly selecting $m$ samples from $\Dtrain$.
      \IF{$B$ includes $m'$ Trojan samples}
      \STATE Generate adversarial examples from $l \cdot m'$ Trojan samples in $B$ by the gate modification attack with TTCD.
      \STATE Add all the adversarial examples to $B$.
      \ENDIF
      \STATE Update $f$ with mini-batch $B$.
      \STATE $i \leftarrow i + 1$
    \ENDWHILE
    \RETURN $\fr = f$
    \end{algorithmic}
  \end{algorithm}
  
\end{figure}

\subsection{Theoretical Analysis}
\label{sub:theoretical_analysis}
We formally discuss the robustness of our proposed method.
We define two types of adversarial examples: \emph{optimal} and \emph{weak} adversarial examples.

\begin{figure}[t]
  \centering
  \includegraphics[width=0.8\columnwidth]{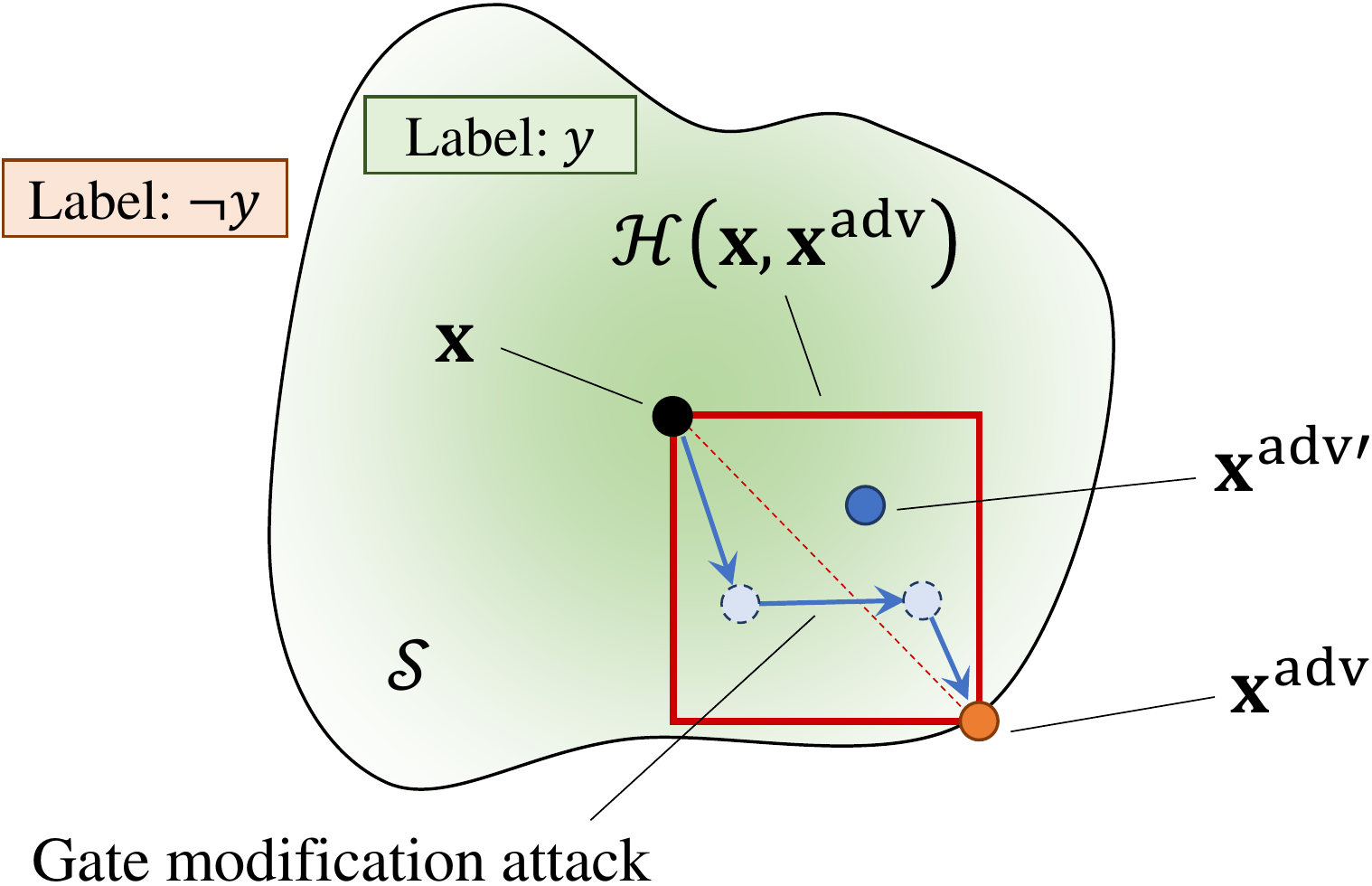}
  \caption{Conceptual illustration of adversarial training. The green shaded area is classified as a Trojan net by the target model.}
  \label{fig:adversarial-example-image}
\end{figure}

\begin{definition}[Optimal adversarial example]
Let $f$ be a classifier. Let $\epsilon$ be a small positive value. We call a solution of the following optimization problem the optimal adversarial example for a given feature vector $\mathbf{x} \in \bbR^{n}$:
\begin{align}
\min_{\mathbf{x}' \in \bbR^{n}} \log f(\mathbf{x}')\ \mathrm{s.t.}\ || \mathbf{x}' - \mathbf{x} ||_\infty < \epsilon. 
\end{align}
\end{definition}

\begin{definition}[Weak adversarial example]
\label{def:weak_adv}
Let $\mathbf{x}^{\mathsf{adv}} \in \bbR^n$ be the optimal adversarial example for a given feature vector $\mathbf{x} \in \bbR^n$.
Let $\mathbf{x}^{\mathsf{adv}'} (\not = \mathbf{x}^
\mathsf{adv}) \in \bbR^n$ be an adversarial example for the same feature vector $\mathbf{x}$.
If $\mathbf{x}^{\mathsf{adv}'}$ satisfies the following conditions, we say that $\mathbf{x}^{\mathsf{adv}'}$ is a weak adversarial example.
\begin{align}
\mathbf{x}^{\mathsf{adv}'} = \mathbf{x} + \boldsymbol{\gamma} (\mathbf{x}^\mathsf{adv} - \mathbf{x}), \mathrm{where\ } \boldsymbol{\gamma} \in [0, 1]^n. 
\end{align}

\end{definition}

Next, we consider a classifier that is robust to adversarial examples.
Let a continuous function $f$ be a classifier constructed with a training dataset $\Dtrain$.
Let $\mathbf{x}_1, \mathbf{x}_2 \in \bbR^{n}$ be samples with the same label $y$ in $\Dtrain$.
Let $\mathcal{H}(\mathbf{x}_1, \mathbf{x}_2)$ be an $n$-dimensional hypercube where $\mathbf{x}_1$ and $\mathbf{x}_2$ are vertices on the longest diagonal.
Let $\mathcal{S}$ be a space that $f$ forms for the label $y$.
% We assume that if $\mathcal{H}$ does not include any samples with different labels from $y$ in $\Dtrain$ and $f$ is well-trained on $\mathbf{x}_1$ and $\mathbf{x}_2$, the following relation holds: $\mathcal{H} \subseteq \mathcal{S}$.
We assume that if $\mathcal{H}$ does not include any samples with different labels from $y$ in $\Dtrain$ and if $f$ is well-trained on $\mathbf{x}_1$ and $\mathbf{x}_2$, then the following relation holds: $\mathcal{H} \subseteq \mathcal{S}$.

\begin{theorem}
\label{thr:robut_model}
Let $\mathbf{x}^\mathsf{adv} \in \bbR^{n}$ be the optimal adversarial example for a given feature vector $\mathbf{x} \in \bbR^{n}$ whose label is $y$.
Let $f$ be a classifier that is well-trained on $\mathbf{x}$ and $\mathbf{x}^\mathsf{adv}$.
If $f$ is a continuous function, $f$ correctly classifies any weak adversarial example $\mathbf{x}^{\mathsf{adv}'} \in \bbR^{n}$ for $\mathbf{x}$.
\end{theorem}

\begin{proof}
By the definition of the adversarial example~\cite{Goodfellow2015}, the distance between $\mathbf{x}^\mathsf{adv}$ and $\mathbf{x}$, $\| \mathbf{x}^\mathsf{adv} - \mathbf{x}\|_\infty$ is small, and thus the probability that the space $\mathcal{H}(\mathbf{x}, \mathbf{x}^\mathsf{adv})$ includes samples with different labels from $y$ is low.
Hence, by Definition~\ref{def:weak_adv} and the above assumption, the theorem holds.
\end{proof}

Fig.~\ref{fig:adversarial-example-image} shows a conceptual illustration of Theorem~\ref{thr:robut_model}.
We assume that TTCD with respect to a Trojan net $\et \in \Et$ is minimized when the value of each element in the feature vector $\xet \in \bbR^{n}$ that represents $\et$ changes the most.
According to the description of each feature~\cite{Hasegawa2017}, when the closest gates to $\et$ are modified, the value of each element in $\xet$ changes the most.
Since the gate modification attack with TTCD selects gates to be modified without any constraints, it can select the closest gates.
Thus, the attack with TTCD produces the optimal adversarial example for $\et$.
On the basis of this assumption, we have the following lemma.

\begin{lemma}
\label{lem:weak_adv_net}
Suppose that an adversarial example generated by the gate modification attack with TTCD for a given Trojan net $\et \in \Et$ is the optimal adversarial example $\xet^\mathsf{adv} = (x^{\mathsf{adv}}_{\et,1}, \ldots, x^{\mathsf{adv}}_{\et,n})\in \bbR^n$ for the net $\et$.
Then, any adversarial example generated by the gate modification attack with $\alpha$-TCD becomes a weak adversarial example $\xet^{\mathsf{adv}'} = (x^{\mathsf{adv}'}_{\et,1}, \ldots, x^{\mathsf{adv}'}_{\et,n}) \in \bbR^n$.
\end{lemma}

\begin{proof}
% By the definitions of modification patterns shown in Fig.~\ref{fig:modification_patterns} including the provided ones in \cite{Nozawa2019a}, gate modification attacks change the value of each element in $\xet$ in only a certain direction.
By the definitions of the modification patterns shown in Fig.~\ref{fig:modification_patterns}, gate modification attacks change the value of each element in $\xet$ in only a certain direction.
Since the gate modification attack with $\alpha$-TCD has specific constraints, the attack may not be able to modify the closest gates to the net $\et$.
Hence we have that if $x_{e_\mathsf{t}, i} \leq x_{e_\mathsf{t}, i}^{\mathsf{adv}}$, then $x_{e_\mathsf{t}, i} \leq x_{e_\mathsf{t}, i}^{\mathsf{adv}'} \leq x_{e_\mathsf{t}, i}^{\mathsf{adv}}$, and otherwise $x_{e_\mathsf{t}, i} \geq x_{e_\mathsf{t}, i}^{\mathsf{adv}'} \geq x_{e_\mathsf{t}, i}^{\mathsf{adv}}$, for any $i \in [n]$.
Therefore, by Definition~\ref{def:weak_adv}, the adversarial example generated with $\alpha$-TCD becomes a weak adversarial example.
\end{proof}

We obtain the following theorem to establish the robustness of R-HTDetector.
\begin{theorem}
\label{thr:rhtdetector}
Detection model $\fr$ constructed via adversarial training with TTCD identifies any adversarial example generated by the gate modification attack with $\alpha$-TCD.
\end{theorem}

\begin{proof}
By Lemma~\ref{lem:weak_adv_net} and Theorem~\ref{thr:robut_model}, the theorem holds.
\end{proof}

Theorem \ref{thr:rhtdetector} means that R-HTDetector overcomes gate modification attacks with $\alpha$-TCD for any value of $\alpha$.

\section{Evaluation}
\label{sec:evaluation}
% In this section, we evaluate the performance of gate modification attacks with $\alpha$-TCD and R-HTDetector against these attacks.
In this section, the gate modification attacks with $\alpha$-TCD and R-HTDetector against these attacks are evaluated.

\subsection{Experimental Setup}
\label{sub:setup}

\begin{table}[t]
    \centering
    \caption{Benchmarks in our experiments.}
    \label{tb:benchmark_list}
    \scalebox{0.95}{
    \begin{tabular}{c|rr||c|rr} \hline
        \multicolumn{3}{c||}{Trust-HUB benchmarks} & \multicolumn{3}{c}{TRIT-TC benchmarks} \\ \hline
        Benchmark & Normal & Trojan & Benchmark & Normal & Trojan \\ \hline
        % RS232-T1000 & 283 & 36 & c2670-T000 & 1,015 & 4 \\ 
        % RS232-T1100 & 284 & 36 & c2670-T001 & 1,017 & 6 \\ 
        % RS232-T1200 & 289 & 34 & c2670-T002 & 1,015 & 4 \\ 
        % RS232-T1300 & 287 & 29 & c3540-T000 & 1,190 & 5 \\ 
        % RS232-T1400 & 273 & 45 & c3540-T001 & 1,191 & 6 \\ 
        % RS232-T1500 & 283 & 39 & c3540-T002 & 1,191 & 6 \\ 
        % RS232-T1600 & 292 & 29 & c5315-T000 & 2,494 & 8 \\ 
        % s15850-T100 & 2,419 & 27 & c5315-T001 & 2,495 & 9 \\ 
        % s35932-T100 & 6,407 & 15 & c5315-T002 & 2,491 & 5 \\ 
        % s35932-T200 & 6,405 & 12 & s1423-T000 & 569 & 4 \\ 
        % s35932-T300 & 6,405 & 37 & s1423-T001 & 571 & 6 \\ 
        % s38417-T100 & 5,798 & 12 & s1423-T002 & 573 & 8 \\ 
        % s38417-T200 & 5,798 & 15 & s13207-T000 & 2,805 & 5 \\ 
        % s38417-T300 & 5,801 & 44 & s13207-T001 & 2,806 & 6 \\ 
        %  &  &  & s13207-T002 & 2,804 & 4 \\ \hline
        RS232-T1000 & 283 & 36 & c2670-T000 & 1011 & 4 \\ 
        RS232-T1100 & 284 & 36 & c2670-T001 & 1011 & 6 \\ 
        RS232-T1200 & 289 & 34 & c2670-T002 & 1011 & 4 \\ 
        RS232-T1300 & 287 & 29 & c3540-T000 & 1185 & 5 \\ 
        RS232-T1400 & 273 & 45 & c3540-T001 & 1185 & 6 \\ 
        RS232-T1500 & 283 & 39 & c3540-T002 & 1185 & 6 \\ 
        RS232-T1600 & 292 & 29 & c5315-T000 & 2486 & 8 \\ 
        s15850-T100 & 2419 & 27 & c5315-T001 & 2486 & 9 \\ 
        s35932-T100 & 6407 & 15 & c5315-T002 & 2486 & 5 \\ 
        s35932-T200 & 6405 & 12 & s1423-T000 & 565 & 4 \\ 
        s35932-T300 & 6405 & 37 & s1423-T001 & 565 & 6 \\ 
        s38417-T100 & 5798 & 12 & s1423-T002 & 565 & 8 \\ 
        s38417-T200 & 5798 & 15 & s13207-T000 & 2800 & 5 \\ 
        s38417-T300 & 5801 & 44 & s13207-T001 & 2800 & 6 \\ 
        &  &  & s13207-T002 & 2800 & 4 \\ \hline
    \end{tabular}
    }
\end{table}

% \begin{table}[t]
%     \centering
%     \caption{Benchmarks used in our experiments.}
%     \label{tb:benchmark_list}
%     \begin{tabular}{ccc||ccc} \hline
%         \multicolumn{3}{c||}{Trust-HUB benchmarks} & \multicolumn{3}{c}{TRIT-TC benchmarks} \\ \hline
%         Series & Trojan nets & Netlists & Series & Trojan nets & Netlists \\ \hline
%         RS232 & 20--50 & 7 & c2670 & less than 10 & 3 \\ 
%         s15850 & $\sim$30 & 1 & c3540 & less than 10 & 3 \\ 
%         s35932 & 10--40 & 3 & c5315 & less than 10 & 3 \\ 
%         s38417 & 10--50 & 3 & s1423 & less than 10 & 3 \\ 
%          &  &  & s13207 & less than 10 & 3 \\ \hline
%     \end{tabular}
% \end{table}

% In the experiments, we use two benchmark sets: the Trust-HUB benchmark~\cite{Salmani2013,shakya2017benchmarking} including 41,024 normal nets and 410 Trojan nets in total and the TRIT-TC benchmark including 24,141 normal nets and 86 Trojan nets in total, as shown in Table~\ref{tb:benchmark_list}.
In the experiments, we use two benchmark sets: the Trust-HUB benchmark, including 14 netlists composed of 41\,024 normal nets and 410 Trojan nets in total, and the TRIT-TC benchmark, including 15 netlists composed of 24\,141 normal nets and 86 Trojan nets in total~\cite{Salmani2013,shakya2017benchmarking,trusthub}, as shown in Table~\ref{tb:benchmark_list}.
Each netlist contains multiple normal and Trojan nets, and Trojan nets are identified by the comments in a netlist source code.
% In addition to the Trust-HUB benchmark, we use the TRIT-TC benchmark netlists to confirm that our R-HTDetector can scale to various circuit designs.
% As mentioned in Section~\ref{sub:htdetection}, we extract 51 features introduced in \cite{Hasegawa2017} for each net in a netlist to identify Trojan nets effectively.
% As mentioned in Section~\ref{sub:htdetection}, we extract 51 features introduced in \cite{Hasegawa2017} for each net in a netlist to identify Trojan nets effectively.
In the experiments, the 51 features listed in Table~\ref{tb:11_features} are extracted for each net in a netlist to identify Trojan nets.

We construct a neural network with three middle layers for the detection model on the basis of \cite{Hasegawa2017}.
% Specifically, input and output layers have 51 and two units, respectively.
Specifically, the number of units in each middle layer is 200, 100, or 50.
The activation function is a sigmoid function, and the optimizer is Adam.

\begin{figure}[t]
    \centering
    \includegraphics[width=1.0\columnwidth]{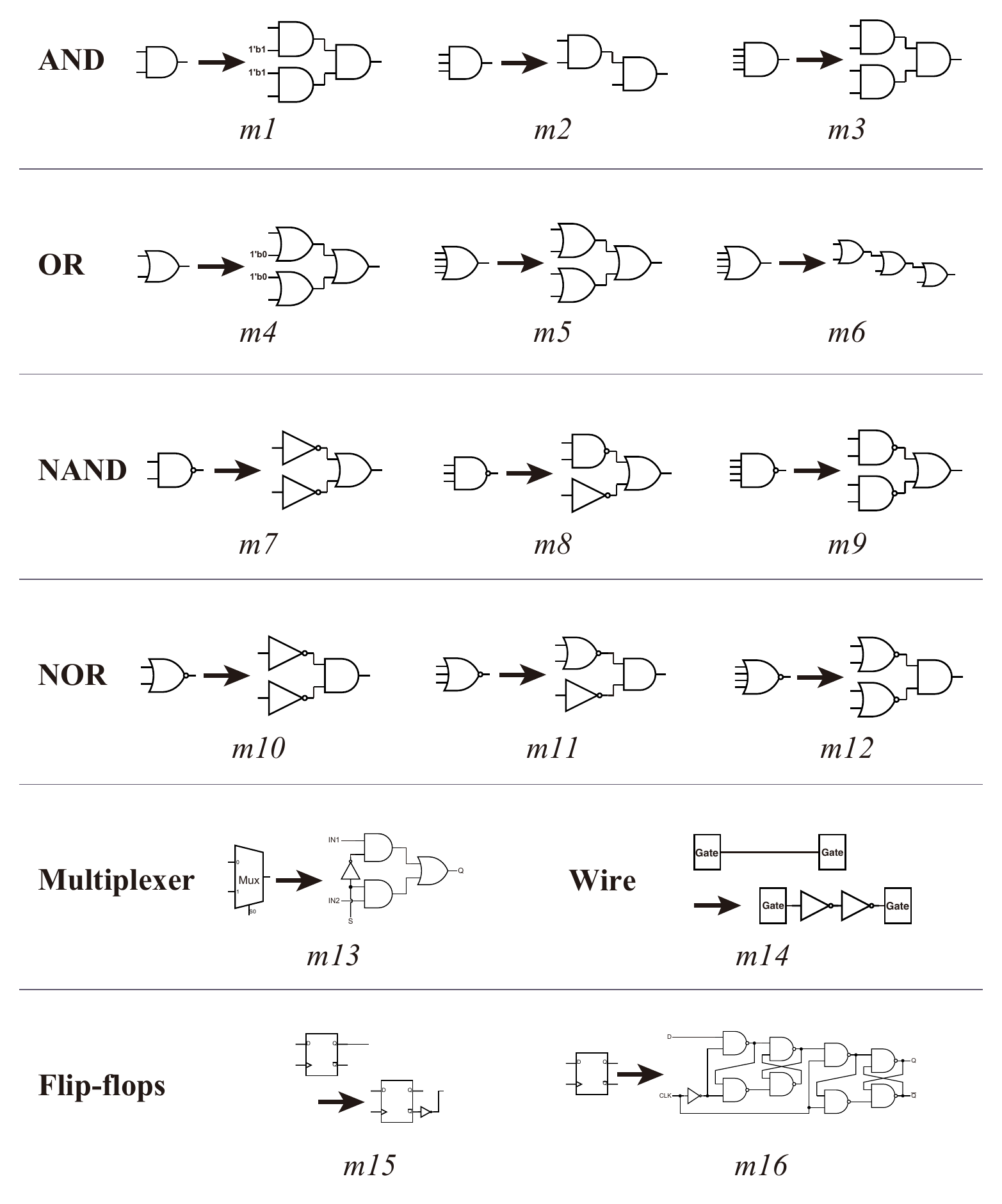}
    % \includegraphics[width=1.0\columnwidth]{fig/modification_patterns.pdf}
    % \caption{Modification patterns. The patterns \textit{m1}--\textit{m12} replace a logic gate with multiple logic gates. The pattern \textit{m13} replaces a multiplexer with a combinatorial circuit. The pattern \textit{m14} replaces a wire with two inverters. The patterns \textit{m15} and \textit{m16} replace a flip-flop with another circuit. Note that we allow the modification pattern \textit{m16} because this modification is only applied to a hardware Trojan where no strict circuit-operation conditions are required.}
    \caption{Modification patterns. \textit{m1}--\textit{m12} replace a logic gate with multiple logic gates. \textit{m13} replaces a multiplexer with a combinatorial circuit. \textit{m14} replaces a wire with two inverters. \textit{m15} and \textit{m16} replace a flip-flop with another circuit. Notably, we allow \textit{m16} because this modification is only applied to an HT in which no strict circuit-operation conditions are required.}
    \label{fig:modification_patterns}
\end{figure}

We use the modification patterns illustrated in Fig.~\ref{fig:modification_patterns}, which are designed based on the gates that compose HTs in the Trust-HUB benchmark.
Other modification patterns, such as the conversion according to De Morgan's laws and adding NOT gates, can also be considered, but we focus on the representative patterns here.

We construct two detection models with and without adversarial training.
We refer to the model constructed with (resp. without) adversarial training as R-HTDetector (resp. normal model).
For R-HTDetector (resp. normal model), we set the epoch size to 10 (resp. 50) and the mini-batch size to $m=16$ (resp. $m=2$).
In the initialization of R-HTDetector, we train the model without adversarial training over one epoch.
We generate adversarial examples for 10\% of the Trojan nets in a mini-batch, i.e., we set $l=0.1$ in the adversarial training.

In the evaluation of both detection models, we assume a scenario in which an adversary generates adversarial examples with respect to the normal model.
Gate modification attacks are performed based on $\alpha$-TCD at $\alpha = $ 1, 2, or $\infty$.
% We evaluate the performance of both models with \emph{true positive rates} (\emph{TPRs}) and \emph{true negative rates} (\emph{TNRs}) as follows: $\mathrm{TPR} = {|\Et \cap Y_\mathsf{t}|} / {|\Et|}$ and $\mathrm{TNR} = {|(E \setminus \Et) \cap Y_\mathsf{n}|} / {|E \setminus \Et|}$, where $Y_\mathsf{t}$ (resp. $Y_\mathsf{n}$) is a set of the nets classified as Trojan nets (resp. normal nets).
We evaluate the performance of both models with TPRs and TNRs defined in Section~\ref{sec:preliminaries}.
When evaluating the performance for a target netlist, we construct the detection model with a training dataset excluding the target netlist (also known as leave-one-out cross-validation).
For instance, when we evaluate the performance of RS232-T1000 using the Trust-HUB benchmark netlists, the detection model is trained with the remaining 13 benchmarks, excluding RS232-T1000.
The trained model is used to classify each net in RS232-T1000.
The Trojan nets are oversampled to balance the training dataset.

\subsection{Experimental Results}
\label{sub:results}

\begin{figure*}[t]
    \centering
    \includegraphics[width=0.7\linewidth]{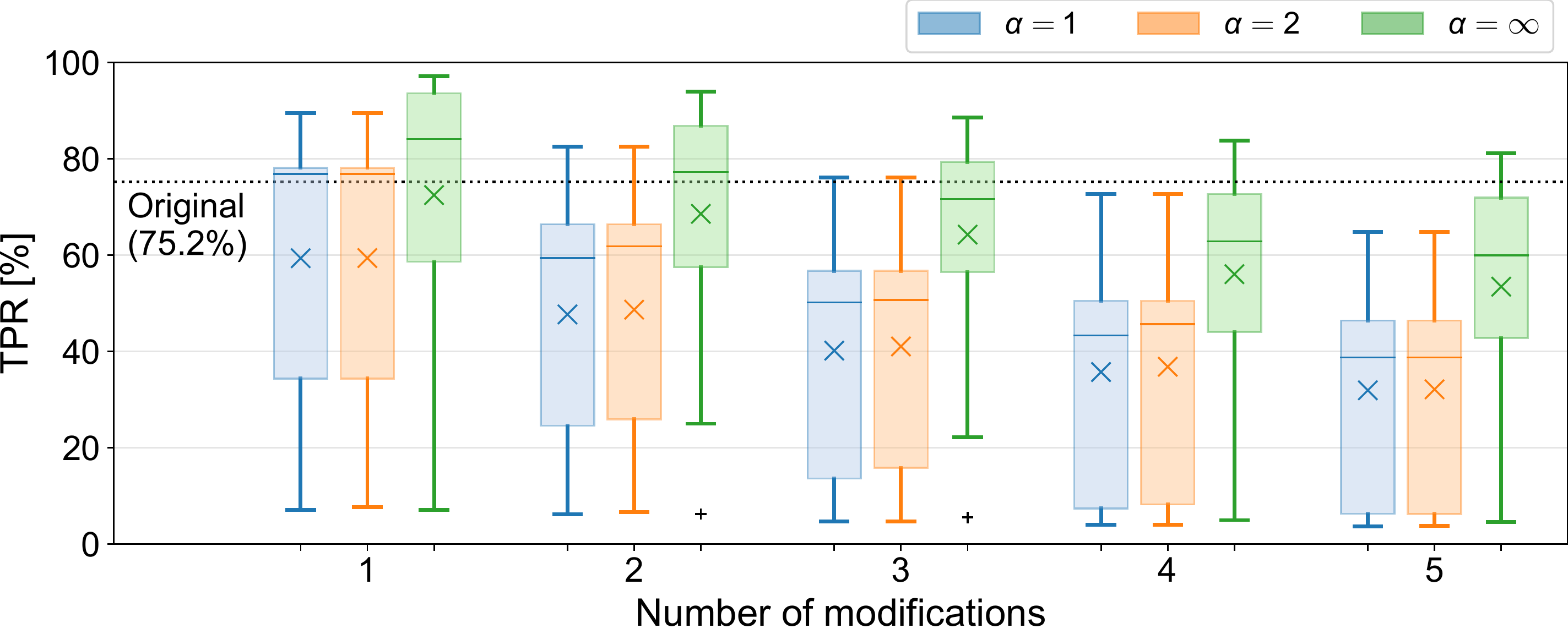}
    \caption{Experimental results on gate modification attacks using the Trust-HUB benchmark, including 14 netlists composed of 41\,024 normal nets and 410 Trojan nets in total.
    The x-axis represents the number of modifications, and the y-axis represents the TPR.
    The horizontal dotted line represents the average TPR for the original netlists. The horizontal line in a boxplot represents the median TPR. The `$\times$' plot denotes an average TPR, and the `+' plot denotes an outlier.}
    \label{fig:attack-results}
\end{figure*}

\begin{table*}[t]
    \centering
    \begin{minipage}{0.96\linewidth}
    \centering
    \caption{Comparison of detection performance between R-HTDetector and normal models (Trust-HUB benchmark).}
    \label{tb:res-trusthub}
    \scalebox{1.00}{
    \begin{tabular}{c||cc|cc||cc|cc|cc} \hline
         & \multicolumn{4}{c||}{Original samples} & \multicolumn{6}{c}{Gate modification-attacked samples (at the fifth modifications)} \\ \cline{2-11}
         & \multicolumn{2}{c|}{TPR} & \multicolumn{2}{c||}{TNR} & \multicolumn{2}{c|}{$\alpha=1$ (TPR)} & \multicolumn{2}{c|}{$\alpha=2$ (TPR)} & \multicolumn{2}{c}{$\alpha=\infty$ (TPR)} \\ \cline{2-11}
        % Benchmarks & Normal & R-HTDetector & Normal & R-HTDetector & Normal & R-HTDetector & Normal & R-HTDetector & Normal & R-HTDetector \\ \hline
        Benchmarks & ~Normal~ & ~R-HTD~ & ~Normal~ & ~R-HTD~ & ~Normal~ & ~R-HTD~ & ~Normal~ & ~R-HTD~ & ~Normal~ & ~R-HTD~ \\ \hline
        RS232-T1000 & 100.0\% & 100.0\% & 98.2\% & 94.3\% & 58.7\% & 100.0\% & 63.0\% & 100.0\% & 78.3\% & 100.0\% \\ 
        RS232-T1100 & 100.0\% & 100.0\% & 96.5\% & 93.3\% & 58.7\% & 100.0\% & 58.7\% & 100.0\% & 60.9\% & 100.0\% \\ 
        RS232-T1200 & 97.1\% & 97.1\% & 98.6\% & 96.2\% & 50.0\% & 90.9\% & 50.0\% & 90.9\% & 54.5\% & 93.2\% \\ 
        RS232-T1300 & 100.0\% & 100.0\% & 98.3\% & 94.8\% & 41.0\% & 100.0\% & 41.0\% & 100.0\% & 35.9\% & 100.0\% \\ 
        RS232-T1400 & 100.0\% & 100.0\% & 99.6\% & 98.2\% & 45.5\% & 100.0\% & 54.5\% & 100.0\% & 54.5\% & 100.0\% \\ 
        RS232-T1500 & 97.4\% & 100.0\% & 98.9\% & 94.3\% & 57.1\% & 95.9\% & 53.1\% & 95.9\% & 77.6\% & 98.0\% \\ 
        RS232-T1600 & 96.6\% & 96.6\% & 98.3\% & 92.1\% & 64.1\% & 97.4\% & 74.4\% & 97.4\% & 61.5\% & 97.4\% \\ 
        s15850-T100 & 48.1\% & 74.1\% & 96.0\% & 93.3\% & 24.3\% & 64.9\% & 24.3\% & 64.9\% & 39.5\% & 60.5\% \\ 
        s35932-T100 & 60.0\% & 80.0\% & 71.3\% & 69.3\% & 34.6\% & 46.2\% & 37.5\% & 54.2\% & 68.0\% & 80.0\% \\ 
        s35932-T200 & 8.3\% & 8.3\% & 100.0\% & 99.9\% & 4.5\% & 9.1\% & 4.5\% & 13.6\% & 4.5\% & 36.4\% \\ 
        s35932-T300 & 73.0\% & 83.8\% & 99.5\% & 99.7\% & 51.0\% & 71.4\% & 51.0\% & 71.4\% & 82.0\% & 88.0\% \\ 
        s38417-T100 & 50.0\% & 66.7\% & 99.9\% & 99.9\% & 4.5\% & 68.2\% & 18.2\% & 72.7\% & 4.5\% & 77.3\% \\ 
        s38417-T200 & 40.0\% & 73.3\% & 100.0\% & 98.7\% & 28.0\% & 76.0\% & 28.0\% & 76.0\% & 4.0\% & 84.0\% \\ 
        s38417-T300 & 81.8\% & 88.6\% & 100.0\% & 99.9\% & 68.5\% & 75.9\% & 68.5\% & 75.9\% & 77.4\% & 83.0\% \\ \hline
        Average & 75.2\% & 83.5\% & 96.8\% & 94.6\% & 42.2\% & 78.3\% & 44.8\% & 79.5\% & 50.2\% & 85.6\% \\ \hline
    \end{tabular}
    }
    \begin{flushleft}
        \footnotesize{* R-HTD: R-HTDetector}
    \end{flushleft}
    \end{minipage}
\end{table*}

The Trust-HUB benchmark is employed for the initial evaluation.
In the evaluation, a gate modification attack is launched on the dataset.
Next, adversarial training is performed to enhance the robustness of the detection model against gate modification-attacked samples.

\smallskip
\noindent
\textbf{Gate modification attack.}
We evaluate the performance of the generalized gate modification attacks described in Section~\ref{sub:alpha_gm_attacks}.
Fig.~\ref{fig:attack-results} shows the results of the attacks with $\alpha$-TCD.
The gate modification-attacked samples are generated and classified using the normal model.
% It can be seen from Fig.~\ref{fig:attack-results} that increasing the number of modifications decreases the average TPRs.
It can be seen that an increase in the number of modifications decreases the average TPRs.
% At the fifth modification, the average TPRs decrease to 50\% or below.
% The trend of degradation in the TPR is different among the values of $\alpha$.
% The TPRs at $\alpha=1$ and $2$ change in a similar way. 
% They decrease from 1 to 3 modifications and then slightly decrease at four or above modifications.
The TPRs at $\alpha=1$ and $2$ change in a similar way, i.e., they decrease with 1--3 modifications and slightly decrease afterward.
% In contrast, the average and median TPRs at $\alpha=\infty$ continuously decrease, although they are larger than those at $\alpha = 1$ and $2$.
In contrast, the average and median TPRs at $\alpha=\infty$ continuously decrease.
% It is expected that the TPR at $\alpha=\infty$ gets closer to the TPRs at $\alpha=1$ and $2$.
% \begin{figure}[t]
%   \centering
%   \includegraphics[width=1.\columnwidth]{fig/loss-RS232-T1400.pdf}
%   \caption{TCD values when generating the gate modification-attacked samples with the Normal model on the RS232-T1400 benchmark. The x-axis shows the number of modifications, and the y-axis shows the $\alpha$-TCD values. The plots show the $\alpha$-TCD values for all samples generated at each modification. The line plot shows the minimum $\alpha$-TCD value.}
%   \label{fig:loss-RS232-T1400}
% \end{figure}
% Algorithm~\ref{alg:gen_aes} slowly finds a modification pattern at $\alpha=\infty$ because $\alpha$-TCD at $\alpha=\infty$ does not observe the loss values of Trojan nets other than the one with the maximum loss value.
% Fig.~\ref{fig:loss-RS232-T1400} shows the $\alpha$-TCD values during the gate modification attack.
% When $\alpha = 1$ or $2$, the $\alpha$-TCD value gradually decreases as we increase the number of modifications.
% However, when $\alpha = \infty$, the $\alpha$-TCD value does not almost change.
% It is considered because the search space is too small with the small number of Trojan nets.
% Nevertheless, the gate-modification attacks successfully degrade the classification performance of the Normal model.
% As mentioned in Section~\ref{sub:alpha_gm_attacks}, it is expected that many Trojan nets can be concealed by increasing the number of modifications.
Since the TPR is most decreased at five modifications, we set $K=5$ for the Trust-HUB benchmark.

% The `Trust-HUB' row in Table~\ref{tb:at-comp} shows the average TPR at the fifth modification.
Table~\ref{tb:res-trusthub} shows the TPR at the fifth modification.
When evaluating the original samples, we use the normal model or R-HTDetector.
When evaluating the gate modification-attacked samples, we generate the gate modification-attacked samples using the normal model (resp. R-HTDetector) and then classify the generated samples with the same model.
% As shown in Table~\ref{tb:at-comp}, TPRs on the gate modification-attacked samples are 50.2\% or below.
% As shown in the `Average' row in Table~\ref{tb:res-trusthub}, TPRs on the gate modification-attacked samples are 50.2\% or below.
As shown in the `Average' row in Table~\ref{tb:res-trusthub}, the maximum TPR for the gate modification-attacked samples is 50.2\%.
Compared to the results on the original samples, adversarial examples are significantly decreased in terms of the TPR.
Based on the results, we argue that adversarial examples effectively distort HT detection based on machine learning.

\begin{figure}[t]
  \centering
  \includegraphics[width=1.\columnwidth]{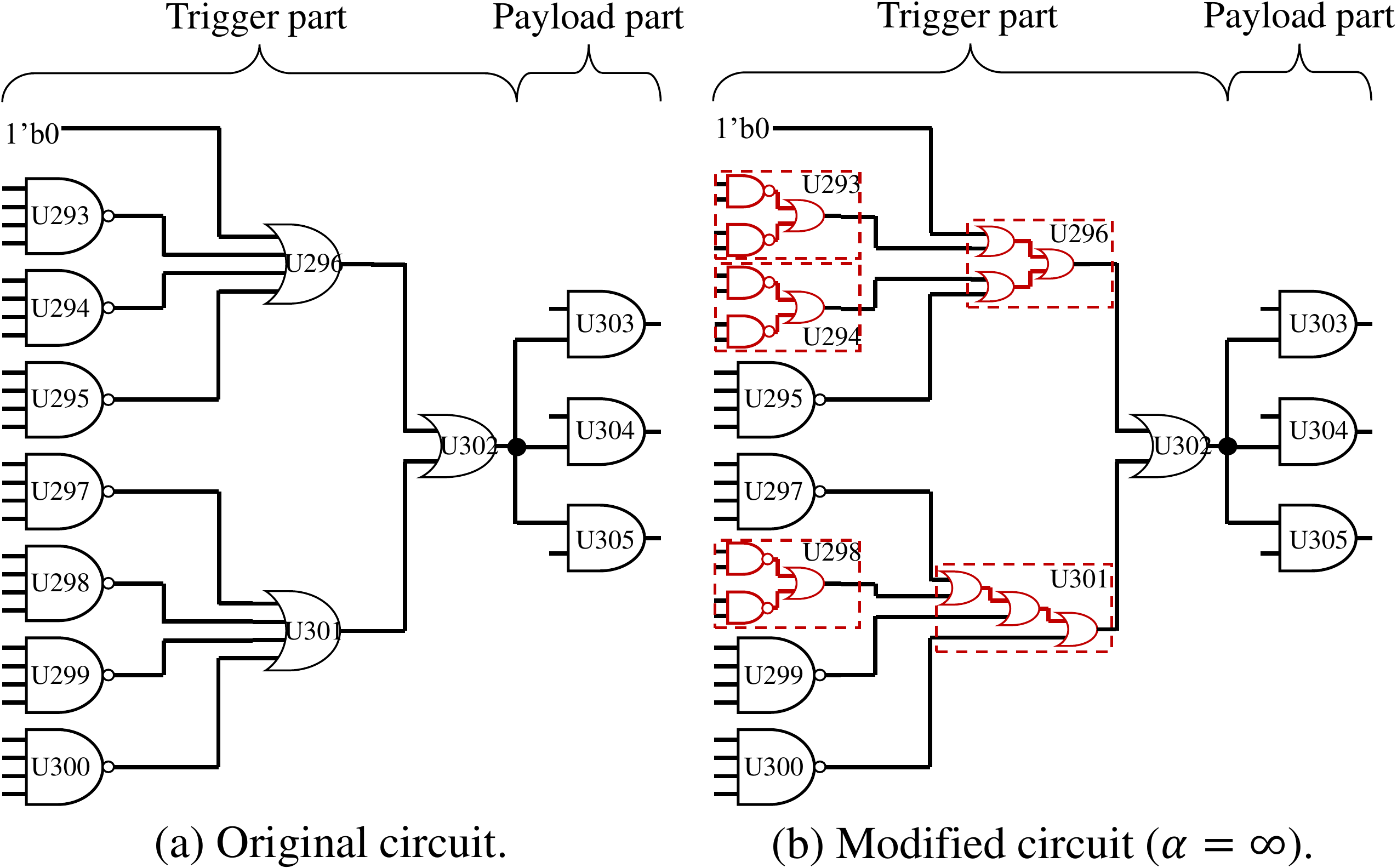}
  \caption{Gate modification attack on RS232-T1000 when $\alpha = \infty$.
  Five gates are replaced with logically equivalent circuits (shown in red).
  }
  \label{fig:RS232-T1000-modification-alpha_inf}
\end{figure}

Fig.~\ref{fig:RS232-T1000-modification-alpha_inf} depicts an example of a gate modification attack.
In the figure, a gate modification attack is launched on RS232-T1000 with $\alpha=\infty$.
Since we set $K=5$, five gates are replaced with logically equivalent circuits.
As shown in the figure, the gates in the trigger part are replaced because the trigger circuits tend to have specific features to HTs to determine rare conditions.
For the features in Table~\ref{tb:11_features}, the number of logic-gate fanins~(\#1--5) strongly correspond to the features.
These feature values are perturbed by increasing the logic levels in the trigger part, which causes misclassification.

\begin{figure*}[t]
  \centering
  \includegraphics[width=0.85\linewidth]{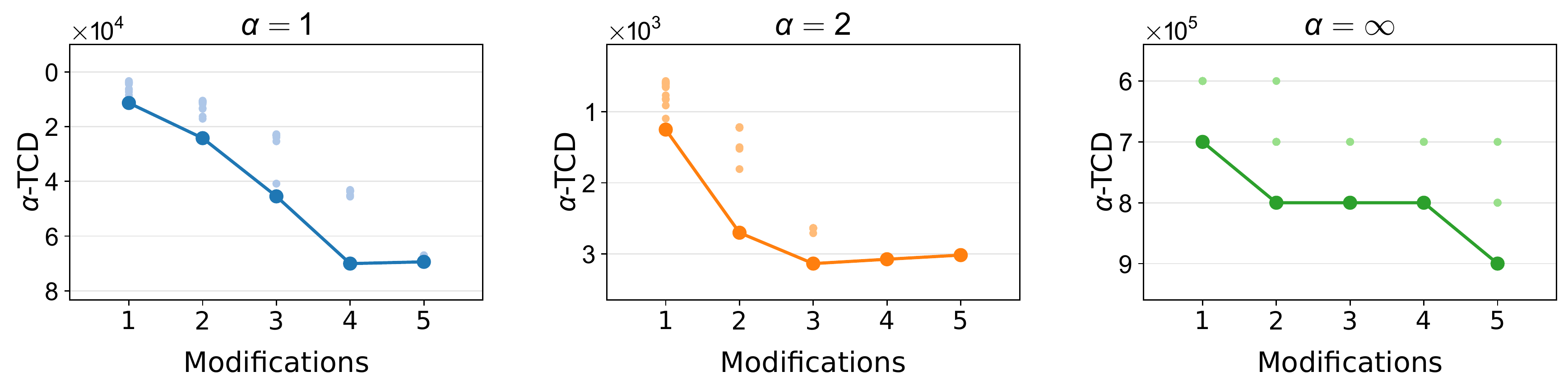}
  \caption{TCD values when generating the gate modification-attacked samples with the normal model on RS232-T1100.
  The x-axis represents the number of modifications, and the y-axis represents the $\alpha$-TCD values.
  The scale of the y-axis is adjusted for each $\alpha$ to determine the trend of $\alpha$-TCD values.
  The plots show the $\alpha$-TCD values for all samples generated at each modification.
  The line plot shows the minimum $\alpha$-TCD value.}
  \label{fig:loss-RS232-T1000}
\end{figure*}
% Algorithm~\ref{alg:gen_aes} slowly finds a modification pattern at $\alpha=\infty$ because $\alpha$-TCD at $\alpha=\infty$ does not observe the loss values of Trojan nets other than the one with the maximum loss value.
Fig.~\ref{fig:loss-RS232-T1000} shows the $\alpha$-TCD values during a gate modification attack.
When $\alpha = 1$ or $2$, the $\alpha$-TCD value gradually decreases as the number of modifications is increased because the value is minimized by considering the whole Trojan nets using the $L^1$ or $L^2$ norm.
The $\alpha$-TCD values converge within the five modifications.
When $\alpha = \infty$, the $\alpha$-TCD value slightly changes and gradually decreases.
Since the value at $\alpha = \infty$ considers the most likely Trojan net, the optimization is performed from a local perspective, resulting in the slight change.
% It is considered that the search space is too small with the small number of Trojan nets when $\alpha = \infty$.

For any $\alpha$, the gate-modification attacks successfully degrade the classification performance of the normal model.
As mentioned in Section~\ref{sub:alpha_gm_attacks}, it is expected that many Trojan nets can be concealed by increasing the number of modifications.
In this sense, the modification attack successfully works on RS232-T1000.

\smallskip
\noindent
\textbf{Adversarial training.}
Next, we evaluate the robustness of R-HTDetector.
We perform adversarial training with TTCD and classify the original and gate modification-attacked samples with $\alpha$-TCD.
% Table~\ref{tb:at-comp} also shows the results of R-HTDetector, i.e., adversarial training on TTCD.
% It can be seen from Table~\ref{tb:at-comp} that R-HTDetector gains 83.5\% average TPR for the original samples, which outperforms the Normal model.
As shown in Table~\ref{tb:res-trusthub}, R-HTDetector gains an average TPR of 83.5\% for the original samples, which outperforms the normal model.
% It can be seen that R-HTDetector gains a higher average TPR for the original netlists than the Normal model.
% Even when the classification performance of the Normal model is not good, our R-HTDetector successfully increases the TPR.
% However, the average TNR for the original samples slightly decreases by 2.2\%, compared to the Normal model.
However, the average TNR for the original samples slightly decreases to 94.6\% from the normal model because adversarial training itself expands the classification area for the Trojan nets, and R-HTDetector does not incorporate adversarial examples for normal nets, as mentioned in Section~\ref{sub:rhtdector}.
Nevertheless, R-HTDetector still maintains a 94.6\% average TNR, which is high enough, and thus, R-HTDetector has satisfactory detection performance for the original samples.

% Table \ref{tb:at-comp} also indicates that the TPRs of the Normal model for gate modification-attacked samples are significantly decreased, compared to the TPR for the original samples.
Table \ref{tb:res-trusthub} also indicates that the TPRs of the normal model for gate modification-attacked samples are significantly decreased compared to the TPR for the original samples.
On the other hand, R-HTDetector achieves more than 78\% of the TPRs, which successfully detects adversarial examples.
From these results, the adversarial training with TTCD can overcome any gate modification attack with $\alpha$-TCD. 

\subsection{Scaling to Other Design Datasets}
% \subsubsection{TRIT-TC benchmarks}
\label{sub:trit-benchmarks}

To evaluate the proposed method using other circuit design, we use the TRIT-TC benchmark, as listed in Table~\ref{tb:benchmark_list}.
We choose the five designs and three Trojan-inserted circuits (-T000, -T001, and -T002) for each design; 15 designs are selected.
It should be noted that the beginning of the circuit name (`c' or `s') represents the type of normal circuit (combinatorial or sequential, respectively) and that trigger circuits of the TRIT-TC benchmark are combinatorial circuits.

% The TRIT-TC benchmark netlists are automatically generated from the Trojan insertion tool~\cite{conf/date/CruzHMB18}.
% We choose the 15 netlists from the benchmark suite.
% The TRIT-TC benchmark is different from the Trust-HUB one with respect to the following points:
This benchmark is different from the Trust-HUB benchmark for the following points:
% 1) the original designs (c2670, c3540, c5315, s1423, and s13207) do not appear in the Trust-HUB benchmark netlists, 2) the hardware-Trojan circuits are different; the TRIT-TC benchmark netlists are automatically generated by a tool, and 3) the cell library is different; for example, a five-input OR gate is used in the TRIT-TC benchmark netlists whereas only less than five-input gates are used in the Trust-HUB benchmark netlists in Table~\ref{tb:benchmark_list}.
% 1) the original designs (c2670, c3540, c5315, s1423, and s13207) do not appear in the Trust-HUB benchmark, 2) the HTs are automatically generated by a tool~\cite{conf/date/CruzHMB18}, and 3) the cell library is different from the Trust-HUB benchmark.
\begin{enumerate}
    \item The original designs (c2670, c3540, c5315, s1423, and s13207) do not appear in the Trust-HUB benchmark.
    \item The HTs are automatically generated by a tool~\cite{conf/date/CruzHMB18}.
    \item The cell library is different; for example, a five-input OR gate is used in the TRIT-TC benchmark netlists whereas only less than five-input gates are used in the Trust-HUB benchmark netlists in Table~\ref{tb:benchmark_list}.
\end{enumerate}
% Moreover, because of the difference of the cell library, there are some logic gates where we cannot apply the modification patterns in Fig.~\ref{fig:modification_patterns}.
Because of the difference in the cell library, some logic gates cannot be replaced based on the modification patterns in Fig.~\ref{fig:modification_patterns}.
% This is the situation that a defender has poor resources or little time to protect their products.
Using the TRIT-TC benchmark netlists, we confirm that our proposed method successfully makes a classifier robust even when we learn another set of circuits.
Most of the experimental setups are the same as those described in Section~\ref{sub:setup}.
% Due to the severe imbalance between normal and Trojan nets in the TRIT-TC benchmark, we set a weight value to the loss function.
% The weight value is calculated by the ratio of Trojan nets to all the nets.
Due to the severe imbalance between normal nets and Trojan nets in the TRIT-TC benchmark, we set a weight value to the loss function calculated by the ratio of Trojan nets to the total number of nets.
In addition, we train the dataset with 15 epochs after the five-epoch initialization step without adversarial training because the TRIT-TC dataset is hard to train due to the imbalanced class distribution.

\begin{figure*}[t]
    \centering
    \includegraphics[width=0.7\linewidth]{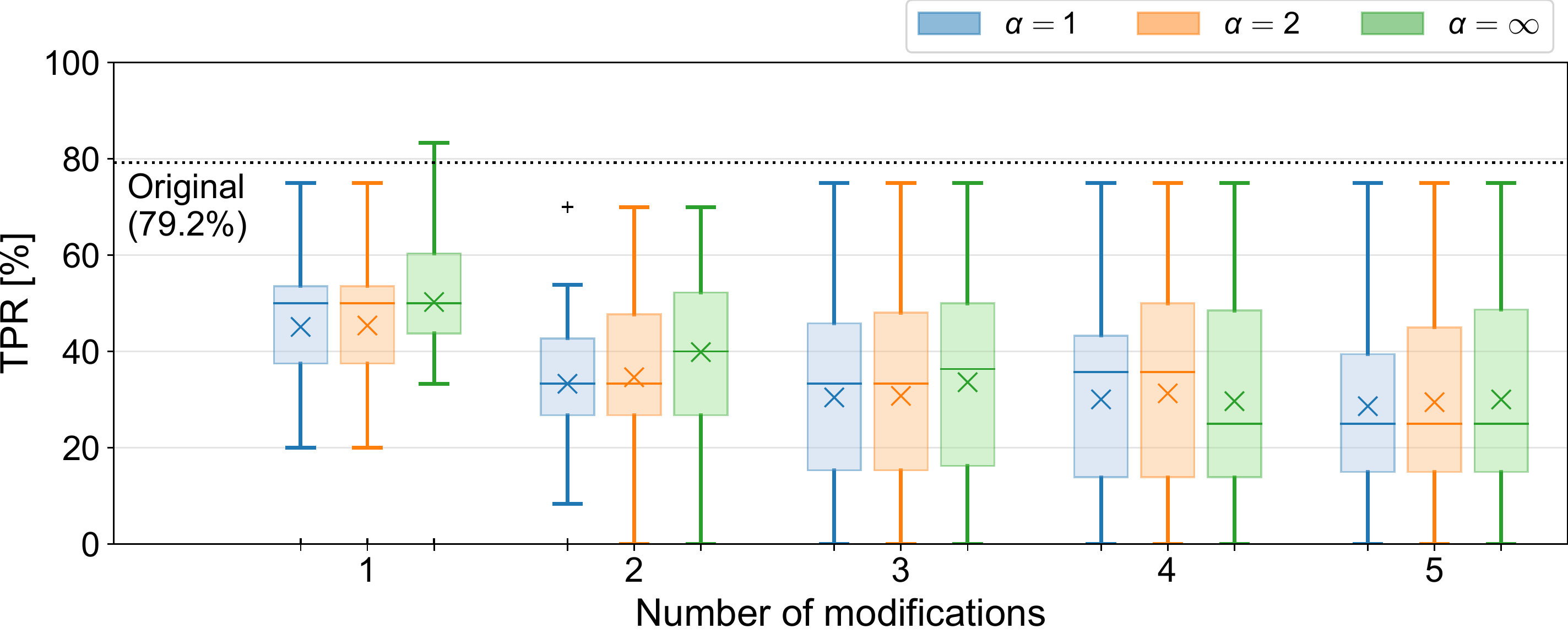}
    \caption{Experimental results for gate modification attacks using the TRIT-TC benchmark, including 15 netlists composed of 24\,141 normal nets and 86 Trojan nets in total.
    The x-axis represents the number of modifications, and the y-axis represents the TPR.
    The horizontal dotted line represents the average TPR for the original netlists. The horizontal line in a boxplot represents the median TPR. The `$\times$' plot denotes an average TPR, and the `+' plot denotes an outlier.}
    \label{fig:attack-results-trit}
\end{figure*}

\begin{table*}[t]
    \centering
    \begin{minipage}{0.96\linewidth}
    \centering
    \caption{Comparison of detection performance between R-HTDetector and normal models (TRIT-TC benchmark).}
    \label{tb:res-trittc}
    \scalebox{1.00}{
    \begin{tabular}{c||cc|cc||cc|cc|cc} \hline
         & \multicolumn{4}{c||}{Original samples} & \multicolumn{6}{c}{Gate modification-attacked samples (at the fifth modifications)} \\ \cline{2-11}
         & \multicolumn{2}{c|}{TPR} & \multicolumn{2}{c||}{TNR} & \multicolumn{2}{c|}{$\alpha=1$ (TPR)} & \multicolumn{2}{c|}{$\alpha=2$ (TPR)} & \multicolumn{2}{c}{$\alpha=\infty$ (TPR)} \\ \cline{2-11}
        % Benchmarks & Normal & R-HTDetector & Normal & R-HTDetector & Normal & R-HTDetector & Normal & R-HTDetector & Normal & R-HTDetector \\ \hline
        Benchmarks & ~Normal~ & ~R-HTD~ & ~Normal~ & ~R-HTD~ & ~Normal~ & ~R-HTD~ & ~Normal~ & ~R-HTD~ & ~Normal~ & ~R-HTD~ \\ \hline
        c2670-T000 & 100.0\% & 100.0\% & 92.6\% & 85.9\% & 40.0\% & 80.0\% & 40.0\% & 80.0\% & 40.0\% & 80.0\% \\ 
        c2670-T001 & 100.0\% & 100.0\% & 93.1\% & 84.0\% & 75.0\% & 83.3\% & 75.0\% & 83.3\% & 75.0\% & 83.3\% \\ 
        c2670-T002 & 50.0\% & 75.0\% & 92.9\% & 90.9\% & 25.0\% & 62.5\% & 25.0\% & 62.5\% & 25.0\% & 62.5\% \\ 
        c3540-T000 & 80.0\% & 100.0\% & 100.0\% & 93.5\% & 11.1\% & 55.6\% & 11.1\% & 55.6\% & 20.0\% & 80.0\% \\ 
        c3540-T001 & 83.3\% & 100.0\% & 99.7\% & 64.6\% & 16.7\% & 83.3\% & 16.7\% & 83.3\% & 16.7\% & 83.3\% \\ 
        c3540-T002 & 83.3\% & 100.0\% & 91.1\% & 68.0\% & 15.4\% & 53.8\% & 15.4\% & 53.8\% & 0.0\% & 38.5\% \\ 
        c5315-T000 & 37.5\% & 87.5\% & 94.3\% & 78.4\% & 18.8\% & 37.5\% & 18.8\% & 37.5\% & 0.0\% & 75.0\% \\ 
        c5315-T001 & 77.8\% & 77.8\% & 92.1\% & 86.3\% & 50.0\% & 62.5\% & 50.0\% & 62.5\% & 62.5\% & 75.0\% \\ 
        c5315-T002 & 80.0\% & 100.0\% & 93.6\% & 71.0\% & 20.0\% & 80.0\% & 20.0\% & 80.0\% & 20.0\% & 80.0\% \\ 
        s1423-T000 & 75.0\% & 100.0\% & 97.7\% & 90.8\% & 50.0\% & 83.3\% & 50.0\% & 83.3\% & 50.0\% & 83.3\% \\ 
        s1423-T001 & 83.3\% & 83.3\% & 98.1\% & 91.9\% & 35.7\% & 92.9\% & 35.7\% & 92.9\% & 35.7\% & 92.9\% \\ 
        s1423-T002 & 37.5\% & 100.0\% & 96.5\% & 86.9\% & 56.3\% & 100.0\% & 53.3\% & 100.0\% & 80.0\% & 100.0\% \\ 
        s13207-T000 & 100.0\% & 100.0\% & 99.4\% & 96.2\% & 0.0\% & 76.9\% & 0.0\% & 76.9\% & 0.0\% & 76.9\% \\ 
        s13207-T001 & 100.0\% & 100.0\% & 98.1\% & 96.1\% & 28.6\% & 64.3\% & 28.6\% & 64.3\% & 21.4\% & 85.7\% \\ 
        s13207-T002 & 100.0\% & 100.0\% & 99.5\% & 95.5\% & 50.0\% & 100.0\% & 50.0\% & 100.0\% & 37.5\% & 100.0\% \\ \hline
        Average & 79.2\% & 94.9\% & 95.9\% & 85.3\% & 32.8\% & 74.4\% & 32.6\% & 74.4\% & 32.3\% & 79.8\% \\ \hline
    \end{tabular}
    }
    \begin{flushleft}
        \footnotesize{* R-HTD: R-HTDetector}
    \end{flushleft}
    \end{minipage}
\end{table*}

\smallskip
\noindent
\textbf{Gate modification attack.}
% Similar to the Section~\ref{sub:results}, we first evaluate the performance of the generalized gate modification attacks described in Section~\ref{sub:alpha_gm_attacks}.
Similar to Section~\ref{sub:results}, we evaluate the generalized gate modification attacks.
Fig.~\ref{fig:attack-results-trit} shows the results for the attacks with $\alpha$-TCD using the normal model.
It can be seen that increasing the number of modifications decreases the average TPRs.
The average TPRs at $\alpha=1$ and $2$ are almost the same, while the average TPR at $\alpha=\infty$ becomes closer to the other TPRs as the number of modifications increases.
% At the third modification, the average TPRs decrease to less than 40\% in any $\alpha$-TCD, and they hardly change at the fourth and fifth modifications.
At the third modification, the average TPRs decrease to less than 40\% in any $\alpha$-TCD.
The average TPRs at the fourth and fifth modifications are similar to those at the third modification because the HT is tiny and the limited number of Trojan nodes can be modified.
Here, we set the number of modifications to $K = 4$ in this experiment because the average TPRs are almost the same as for $K = 5$.
% Note that the boxplots in Fig.~\ref{fig:attack-results-trit} contain several outliers.
% Because of the tiny hardware Trojans in the TRIT-TC benchmark netlists, a small misclassification can greatly affect the TPR, resulting in such outliers.
% Those at the first modification come from s13207-T000, and that at the second modification comes from c2670-T001.
% The hardware Trojans in s13207 designs are tiny, and it is easy to misclassify them.
% On the contrary, the c2670 design is a small circuit compared to s13207, and therefore it is easy to detect hardware Trojans even if some parts of them are modified.

\begin{figure}[t]
    \centering
    \includegraphics[width=1.\columnwidth]{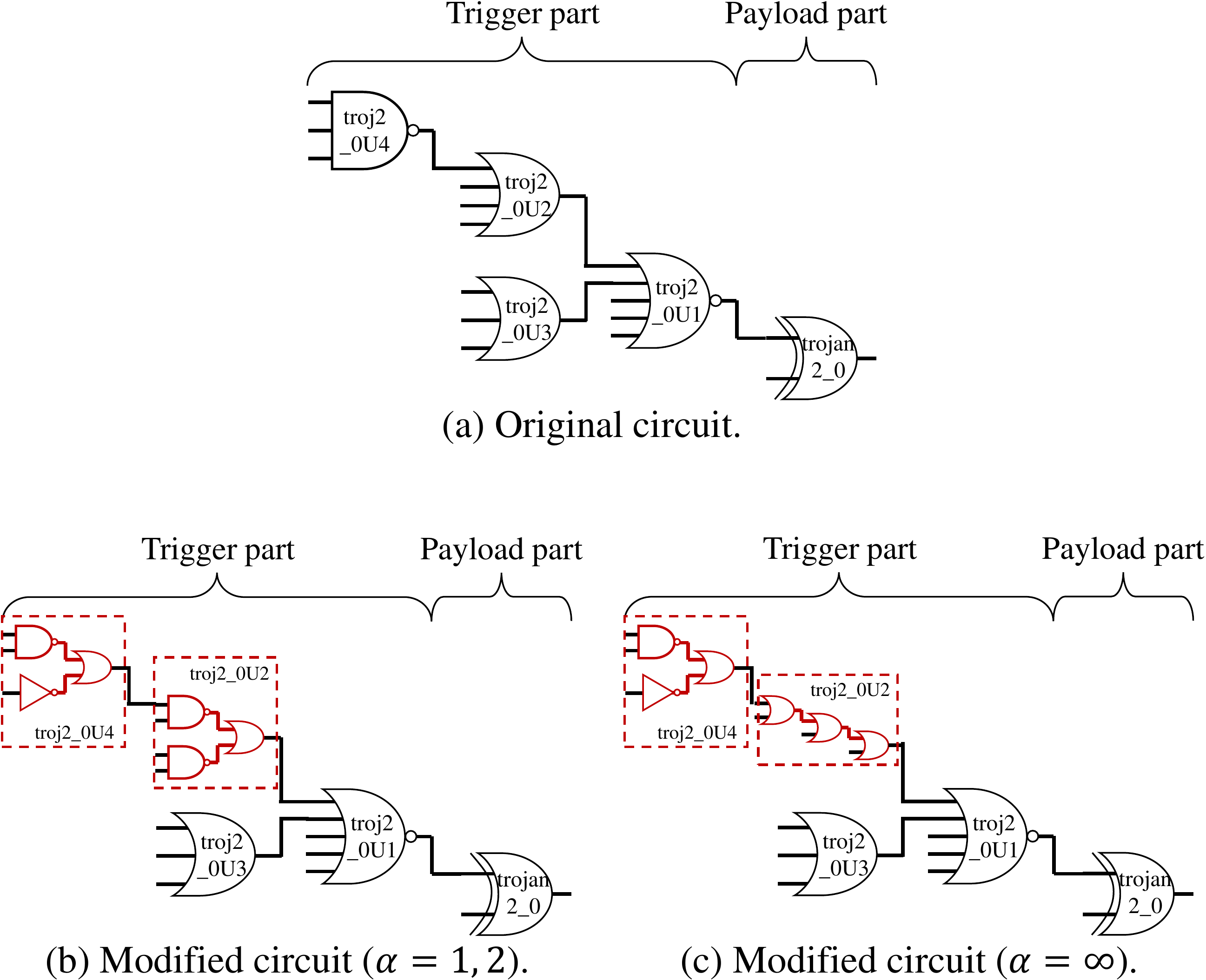}
    \caption{Gate modification attack on s13207-T002.
    Two gates are replaced with logically equivalent circuits (shown in red).
    }
    \label{fig:s13207-T002-modification}
\end{figure}

The effect of $\alpha$ for the gate modification attack is different between the Trust-HUB benchmarks and the TRIT-TC benchmarks.
Table~\ref{tb:res-trittc} shows the experimental results for the TRIT-TC benchmark.
The sixth, eighth, and tenth columns show the results of the gate modification attack for the normal model.
Although $\alpha$-TCD at $\alpha=\infty$ becomes the highest average TPR in Trust-HUB, the metric achieves the lowest average TPR in TRIT-TC.
It is considered that there are a few modifiable gates in the TRIT-TC benchmark, and $\alpha=\infty$ can fit in such a situation.
Fig.~\ref{fig:s13207-T002-modification} shows the gate modification attack on s13207-T002.
The two modifiable gates are replaced based on the modifications patterns.
Gate \texttt{troj2\_0U2} is replaced with pattern \textit{m5} when $\alpha = 1$ or $2$, whereas the gate is replaced with pattern \textit{m6} when $\alpha = \infty$.
Since pattern \textit{m6} increases the number of logic levels more than pattern \textit{m5}, the features of more Trojan nets are modified.
As a result, the detection rate is further decreased when $\alpha = \infty$.

\smallskip
\noindent
\textbf{Adversarial training.}
Next, we evaluate the robustness of R-HTDetector using the TRIT-TC benchmark.
% Table~\ref{tb:at-trit-comp} shows the results.
% The last row in Table~\ref{tb:at-comp} shows the results.
The last row in Table~\ref{tb:res-trittc} shows the results.
Similar to the Trust-HUB benchmark, R-HTDetector gains an average TPR of 94.9\% for the original samples, which outperforms the normal model.
% Even when we see the TPR for each benchmark, R-HTDetector gains equal or higher TPR than the Normal model.
However, R-HTDetector decreases the average TNR compared to the normal model.
% The reason for this result is the same as the Trust-HUB benchmark experiment; i.e., R-HTDetector does not incorporate adversarial examples for normal nets.
The reason for this result is the same conclusion reached after the Trust-HUB benchmark experiment; i.e., R-HTDetector does not incorporate adversarial examples for normal nets.
% Furthermore, the gate modification-attacked samples intrinsically have similar feature values to the normal nets due to the tiny HTs.
Furthermore, the HTs in the TRIT-TC benchmark tend to have feature values similar to the normal nets because of the minute HT size.
% Setting the weight to the loss function in this experiment puts much penalty when misclassifying Trojan nets, which expands the classification area for Trojan nets.
% Setting the weight to the loss function in this experiment expands the classification area for Trojan nets and avoids misclassifying them.
Although it is difficult to balance the TPR and TNR, the TNR can be improved by incorporating adversarial examples for normal nets or by changing the weight of the loss function.
% Because of the weight to the loss function in this experiment, R-HTDetector put much more penalty when misclassifying Trojan nets.
% Therefore, R-HTDetector misclassifies confusing normal nets as Trojan nets to avoid missing Trojan nets.
% Therefore, R-HTDetector expands the classification area for Trojan nets to avoid missing them.
% Nevertheless, R-HTDetector still keeps high TNRs that are more than 85\%, and thus R-HTDetector achieves good detection performance for original samples.
Nevertheless, R-HTDetector still retains more than 85\% of the average TNR.
Thus, R-HTDetector achieves sufficient detection performance for the original samples.

% Table \ref{tb:at-trit-comp} also shows that the TPRs of the Normal model using gate modification-attacked samples are significantly decreased compared to those using the original samples.
% Table \ref{tb:at-comp} also shows that the TPRs with the Normal model for gate modification-attacked samples are significantly decreased.
Table \ref{tb:res-trittc} also shows that the TPRs with the normal model for the gate modification-attacked samples are significantly decreased.
Although the average TPR is 79.2\% with the normal model for the original samples, it decreases to lower than 33\% in all $\alpha$ cases for the gate modification-attacked samples.
However, R-HTDetector successfully detects the attacked samples, and the average TPRs are recovered to grater than 74\%.
% It is worth noting that R-HTDetector recovers the TPR for c3540-T002, c5315-T000, and c13207-T000 benchmarks, where the Normal model can detect no Trojan nets due to the gate modification attacks.
Although the TRIT-TC benchmark netlists are different from the Trust-HUB netlists as mentioned at the top of this subsection, R-HTDetector is successful.

From these results, the adversarial training with TTCD successfully overcomes any gate modification attack with $\alpha$-TCD using various benchmarks.

\smallskip
We conclude the experiments of R-HTDetector using the two benchmarks as follows:
First, R-HTDetector increases the average TPR for the original samples, which successfully expands the classification area of HTs, including the gate modification-attacked samples.
The TNRs are still more grater 85\% even when the HTs are minute.
Second, R-HTDetector successfully recovers TPRs for gate modification-attacked samples with any $\alpha$-TCD attack.
Even when the cell library or circuit designs vary, the R-HTDetector successfully works.

\section{Related Works}

This section presents several related works on this paper.

\subsection{Other HT Detection Models}

Machine learning-based HT detection methods that target other than gate-level netlists have been researched~\cite{8952724}.
Numerous feature types learned by machine learning models were proposed.
In side channel-based HT detection methods, path-delay and power consumption are often employed as the features representing the target circuit.
Such feature values can be easily perturbed by the proposed method.
Although the proposed method focuses on the structural features of gate-level netlists, the key idea is to add perturbation to the trained features by replacing a logic gate with a set of logically equivalent circuits.
In this sense, the idea can be applied to the method that utilizes the feature values that can be altered by logically equivalent modification.

Switching probability-based and similar approaches are not addressed in this paper because logically equivalent modification does not change the switching probability.
Specifically, the SCOAP value-based approach~\cite{7577739} does not affect the perturbation by a gate modification attack.
However, the approach may not be applicable to specific application circuits~\cite{conf/ismvl/ItoUH21}.
In terms of the adversarial example attacks with the approach, adversaries may perturb the switching probability of trigger circuits, which makes it difficult for machine learning models to determine the threshold of HT detection.
One solution to constructing a robust model against adversarial examples would be to combine different approaches, such as the SCOAP value-based and structural feature-based approaches.

\subsection{Adversarial Examples on the Tabular Dataset}

The feature values listed in Table~\ref{tb:11_features} are structured as a \emph{tabular} dataset.
Methods of adversarial examples on tabular datasets were proposed such as presented in \cite{info12090375} and \cite{conf/aaai/CartellaAFYAE21}.
The methods synthesize adversarial examples on tabular datasets by using generative adversarial network~(GAN) or genetic algorithms.
In gate-level netlists, feature extraction is an irreversible process.
It is extremely difficult to reproduce a circuit from feature values.
Therefore, generating adversarial examples by directly perturbing feature values is not reasonable.

\subsection{Adversarial Examples on the Graph Dataset}

Graph learning is a growing research area.
The methods of HT detection using graph learning have recently been proposed~\cite{9474174,muralidhar2021contrastive,arxiv:2112.02213v1}.
These methods learn the structural features of HTs, and thus, the proposed method can be applied.

Several adversarial example attack and defense methods on graph data were proposed~\cite{10.1145/3447556.3447566}.
Most of them modify the graph data by adding and/or removing edges and/or nodes.
Such manipulation may destroy the functionality of original circuits; therefore, it is not applicable to gate-level netlists.
To establish a more sophisticated method of generating adversarial examples on gate-level netlists, it is necessary to ensure that the generated example properly works as a circuit after modification of the graph data.

\section{Conclusion}
\label{sec:conclusion}
% We focused on \emph{gate modification attacks} against HT detection at gate label with machine learning, and proposed a new HT detection method, \emph{R-HTDetector}.
% Gate modification attacks generate \emph{adversarial examples} to make a target detection model misclassify Trojan nets.
% In order to defend the model from adversarial examples, we constructed the model on the basis of \emph{adversarial training} specific to HT detection.

% In this paper, we first generalized gate modification attacks for realizing the attacks with various purposes.
% We next established that R-HTDetector is robust to any gate modification attacks from a theoretical point of view.
% We finally demonstrated through experiments that the generalized gate modification attacks significantly degrade the performance of the detection model constructed without adversarial training.
% We also showed that R-HTDetector overcomes any gate modification attacks while keeping the original accuracy. 
% R-HTDetector successfully works for the TRIT-TC benchmark as well as the Trust-HUB benchmark.

% We focused on \emph{gate modification attacks} against HT detection at the gate label with machine learning and proposed a new HT detection method, \emph{R-HTDetector}.
This paper presents \emph{gate modification attacks} against HT detection at the gate label with machine learning and a new HT detection method, \emph{R-HTDetector}.
We first generalized gate modification attacks for realizing attacks with various purposes.
Then we established that R-HTDetector is robust to any gate modification attack from a theoretical point of view.
We demonstrated through experiments that generalized gate modification attacks significantly degrade the performance of the detection model without adversarial training.
We also showed that R-HTDetector overcomes any gate modification attack while maintaining the original accuracy. 
% R-HTDetector successfully works for the TRIT-TC benchmark as well as the Trust-HUB benchmark.

% In the future, we will apply our adversarial training to other advanced machine learning models, and confirm the effectiveness of our approach.
% Furthermore, enhancing the classification performance considering the balance between TPRs and TNRs, especially when the target hardware Trojan is tiny, is another future work.
% Furthermore, we will enhance the classification performance balancing TPRs and TNRs.
In the future, we will apply our adversarial training to other advanced machine learning models.
Additionally, we will enhance the classification performance by balancing the TPRs and TNRs.

% We will further evaluate the robustness of R-HTDetector with different parameter settings (e.g. the number of gate modifications and the number of adversarial examples incorporated into adversarial training).
% We will also compare our proposed method with other robust learning methods such as data augmentation.

% ================= End of contents
%\section*{Acknowledgment}

\bibliographystyle{IEEEtran}
\bibliography{References}

\end{document}